\documentclass[12pt]{article}

\usepackage{amsmath,amsthm}

\usepackage{tikz}
\usetikzlibrary{arrows,matrix,positioning}

\usepackage{amssymb,latexsym}
\usepackage[all]{xy}
\usepackage{booktabs}
\usepackage{array,fullpage,verbatim}

\newtheorem{theorem}{Theorem}[section]
\newtheorem{lemma}[theorem]{Lemma}
\newtheorem{corollary}[theorem]{Corollary}
\newtheorem{proposition}[theorem]{Proposition}

\theoremstyle{definition}
\newtheorem{definition}[theorem]{Definition}

\newtheorem{example}[theorem]{Example}

\DeclareMathOperator{\rk}{rank}
\DeclareMathOperator{\row}{row}

\newcommand{\dt}{\cdot}

\newcommand{\tl}{\pmb{\circ}}
\newcommand{\zz}{\cdot}
 
\begin{document}

\title{Matrix Theory for Minimal Trellises} 
\author{Iwan M. Duursma \thanks{University of Illinois at Urbana-Champaign, Dept of Mathematics and Coordinated Science Laboratory. Partially supported by a grant from the Simons Foundation (\#280107). \rm{duursma@math.uiuc.edu}.} }

\date{September 28, 2015}

\maketitle

\begin{abstract}
Trellises provide a graphical representation for the row space of a matrix. 
The product construction of Kschischang and Sorokine builds minimal conventional trellises from matrices in minimal span form. Koetter and Vardy showed that  minimal tail-biting trellises can be obtained by applying the product construction to submatrices of a characteristic matrix. 

We introduce the unique reduced minimal span form of a matrix and we obtain an expression for the unique reduced characteristic matrix.
Among new properties of characteristic matrices we prove that characteristic matrices are in duality if and only if they have orthogonal column spaces, and that the transpose of a characteristic matrix is again a characteristic matrix if and only if the characteristic matrix is reduced. These properties have clear interpretations for the unwrapped unit memory convolutional code of a tail-biting trellis, they explain the duality for the class of Koetter and Vardy trellises, and they give a natural relation between the characteristic matrix based Koetter-Vardy construction and the displacement matrix based Nori-Shankar construction. 

For a pair of reduced characteristic matrices in duality, one is lexicographically first in a forward direction and the other is lexicographically first in the reverse direction. This confirms a conjecture by Koetter and Vardy after taking into account the different directions for the lexicographical ordering.

\end{abstract}

\section{Introduction} \label{S:intr}

A trellis is a directed graph with edge-labeled paths that represent the possible output sequences of an encoder. 
Trellises facilitate the decoding, via shortest path algorithms, of sequences that have been received over a noisy communication channel. Important examples of such algorithms are the Viterbi maximum likelihood decoder and the BCJR maximum a posteriori decoder. 
For linear block codes, encoding is given by matrix multiplication and a trellis represents the vectors in the row space of the matrix.  

We assume throughout that trellises have no multiple edges and in particular that a path in the trellis is uniquely determined by specifying the vertices along the path.
A conventional trellis comes with a partition $V_0 \cup V_1 \cup \cdots \cup V_n$ of the vertices, such that both $V_0$ and $V_n$ are singletons. In the more general setting of tail-biting trellises, $V_0$ and $V_n$ are in bijection. 
A path $(v_0,v_1,\ldots,v_n)$ is tail-biting if the bijection between $V_0$ and $V_n$ matches the head $v_n$ of the path to its tail $v_0$. It is common use to identify $V_0$ and $V_n$ and to refer to a tail-biting path as a cycle with origin $v_0$ or cycle for short.
A tail-biting path with edge labels $c_0,c_1,\ldots,c_{n-1}$ is said to represent the row vector $(c_0,c_1,\ldots,c_{n-1})$. 

\begin{figure}
\begin{center}
\begin{tikzpicture}
[scale=.5,auto=center,every node/.style={minimum size=0em}]


\newcommand{\y}{0}
\newcommand{\x}{2}
\newcommand{\z}{4}
\node[] at (\z+4,\y) {$V_0$};
\node[] at (\z+8,\y) {$V_1$};
\node[] at (\z+12,\y) {$V_2$};
\node[] at (\z+18,\y) {$V_{i}$};
\node[] at (\z+22,\y) {$V_{i+1}$};
\node[] at (\z+28,\y) {$V_{n-1}$};
\node[] at (\z+32,\y) {$V_n$};

\draw [dashed] (\z+2.9,\y+1) -- (\z+5.0,\y+1) -- (\z+5.0,\y+5) -- (\z+2.9,\y+5) -- (\z+2.9,\y+1); 
\draw [dashed]  (\z+6.9,\y+1) -- (\z+9.0,\y+1) -- (\z+9.0,\y+5) -- (\z+6.9,\y+5) -- (\z+6.9,\y+1); 
\draw [dashed]  (\z+10.9,\y+1) -- (\z+13.0,\y+1) -- (\z+13.0,\y+5) -- (\z+10.9,\y+5) -- (\z+10.9,\y+1); 
\draw [dashed]  (\z+16.9,\y+1) -- (\z+19.0,\y+1) -- (\z+19.0,\y+5) -- (\z+16.9,\y+5) -- (\z+16.9,\y+1); 
\draw [dashed]  (\z+20.9,\y+1) -- (\z+23.0,\y+1) -- (\z+23.0,\y+5) -- (\z+20.9,\y+5) -- (\z+20.9,\y+1); 
\draw [dashed]  (\z+26.9,\y+1) -- (\z+29.0,\y+1) -- (\z+29.0,\y+5) -- (\z+26.9,\y+5) -- (\z+26.9,\y+1); 
\draw [dashed]  (\z+30.9,\y+1) -- (\z+33.0,\y+1) -- (\z+33.0,\y+5) -- (\z+30.9,\y+5) -- (\z+30.9,\y+1); 

\node (s0) at (\z+4,\y+3) {$\bullet$}; \node [below] at (\z+4,\y+2.8) {$v_0$};
\node (s1) at (\z+8,\y+3) {$\bullet$}; \node [below] at (\z+8,\y+2.8) {$v_1$};
\node (s2) at (\z+12,\y+3) {$\bullet$}; \node [below] at (\z+12,\y+2.8) {$v_2$};
\node (s4) at (\z+18,\y+3) {$\bullet$}; \node [below] at (\z+18,\y+2.8) {$v_{i}$};
\node (s5) at (\z+22,\y+3) {$\bullet$}; \node [below] at (\z+22,\y+2.8) {$v_{i+1}$};
\node (s7) at (\z+28,\y+3) {$\bullet$}; \node [below] at (\z+28,\y+2.8) {$v_{n-1}$};
\node (s8) at (\z+32,\y+3) {$\bullet$}; \node [below] at (\z+32,\y+2.8) {$v_n$};

 \path
(s0) edge [->,line width=1.5pt,color=black] (s1) 
(s1) edge [->,line width=1.5pt,color=black] (s2)
(s2) edge [->,line width=1.5pt,color=black,dashed] (s4)
(s4) edge [->,line width=1.5pt,color=black] (s5)
(s5) edge [->,line width=1.5pt,color=black,dashed] (s7)
(s7) edge [->,line width=1.5pt,color=black] (s8);

\node [above] at (\z+6,\y+3.2) {$c_0$};
\node [above] at (\z+10,\y+3.2) {$c_1$};
\node [above] at (\z+20,\y+3.2) {$c_i$};
\node [above] at (\z+30,\y+3.2) {$c_{n-1}$};
\end{tikzpicture}
\end{center}
\caption{Trellis path for a vertex partition $V_0 \cup V_1 \cup \cdots \cup V_n$}
\end{figure}
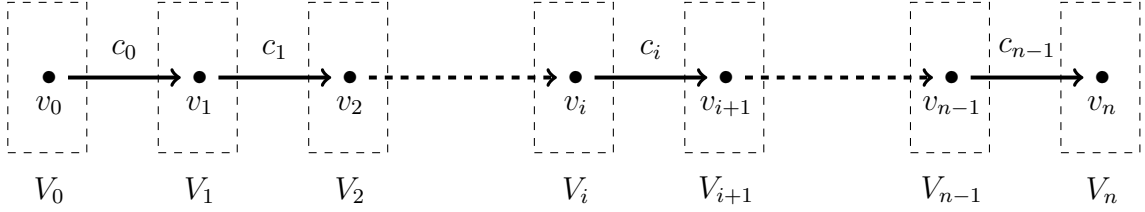

Let $T \subset V_0 \times V_1 \times \cdots \times V_n$ be the full collection of tail-biting paths in a trellis. A trellis is said to represent the row space $C$ of a matrix if each row vector represented by a path in $T$ belongs to $C$ and if each row vector in $C$ is represented by a path in $T$. The representation is one-to-one if each row vector in $C$ is represented by a unique path in $T$.
The trellis construction problem is to find minimal trellises that represent the row space of a given matrix one-to-one while minimizing the vertex count or any other relevant complexity measure for the trellis. 

\subsection*{Constructions for linear trellises}

All trellises considered in this paper are linear trellises. For a linear trellis that represents the row space of a matrix with coefficients in a field $F$, the vertex subsets $V_i$ are $F$-linear spaces. A tail-biting path has a vertex labeling $\nu=(v_0,v_1,\ldots,v_n) \in T$ and an edge labeling $c=(c_0,c_1,\ldots,c_{n-1}) \in C$. The vertex-edge labelings $(\nu,c)$ form a linear subspace $S(T) \subset T \times C$ called the label code of the trellis. A linear trellis is determined by specifying a basis for the label code, that is by specifying the vertex-edge labelings for a subset of paths that generates the trellis as $F$-linear space. Two trellises are linearly isomorphic if there exist linear isomorphisms $\alpha_i : V_i \longrightarrow V'_i$, for $0 \leq i \leq n$, that are compatible with edges labelings, i.e. that map edges $(v_i,c_i,v_{i+1})$ to edges $(\alpha_i(v_i),c_i,\alpha_{i+1}(v_{i+1}))$. Figure~\ref{F:exF3} illustrates for two different trellis constructions how they lead to different but isomorphic trellises (over the ternary field $F= \{0,1,2\}$). 
We briefly describe each of the two constructions.

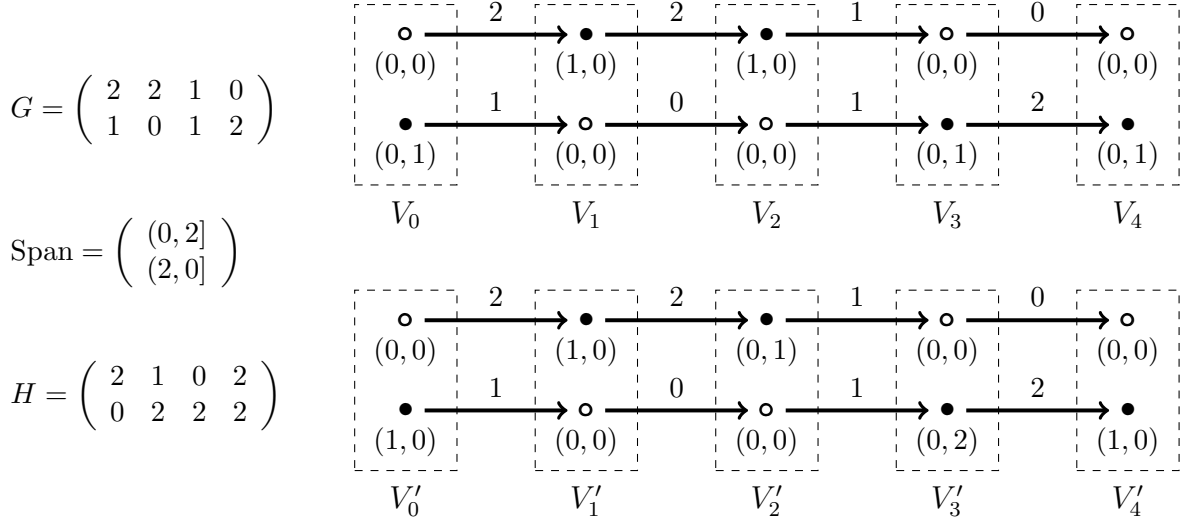
\begin{figure} 
\begin{center}
\begin{tikzpicture}
[scale=.4,auto=center,every node/.style={minimum size=0em}]

\newcommand{\y}{-3}
\newcommand{\x}{2}
\newcommand{\z}{3.5}

\node[right] at (-4,\y+2) {\small $G =
\left( \begin{array}{cccc} 2 &2 &1 &0 \\ 1 &0 &1 &2  \end{array} \right)$};

\node[] at (\z+6,\y-1.5) {$V_0$};
\node[] at (\z+12,\y-1.5) {$V_1$};
\node[] at (\z+18,\y-1.5) {$V_2$};
\node[] at (\z+24,\y-1.5) {$V_3$};
\node[] at (\z+30,\y-1.5) {$V_4$};

\draw [dashed] (\z+4.3,\y-.5) -- (\z+7.7,\y-0.5) -- (\z+7.7,\y+5.5) -- (\z+4.3,\y+5.5) -- (\z+4.3,\y-.5); 
\draw [dashed]  (\z+10.3,\y-.5) -- (\z+13.7,\y-.5) -- (\z+13.7,\y+5.5) -- (\z+10.3,\y+5.5) -- (\z+10.3,\y-.5); 
\draw [dashed]  (\z+16.3,\y-.5) -- (\z+19.7,\y-.5) -- (\z+19.7,\y+5.5) -- (\z+16.3,\y+5.5) -- (\z+16.3,\y-.5); 
\draw [dashed]  (\z+22.3,\y-.5) -- (\z+25.7,\y-.5) -- (\z+25.7,\y+5.5) -- (\z+22.3,\y+5.5) -- (\z+22.3,\y-.5); 
\draw [dashed]  (\z+28.3,\y-.5) -- (\z+31.7,\y-.5) -- (\z+31.7,\y+5.5) -- (\z+28.3,\y+5.5) -- (\z+28.3,\y-.5); 

\node (s0) at (\z+6,\y+4.5) {$\tl$}; \node [below] at (\z+6,\y+4.3) {\small $(0,0)$};
\node (s1) at (\z+12,\y+4.5) {$\bullet$}; \node [below] at (\z+12,\y+4.3) {\small $(1,0)$};
\node (s2) at (\z+18,\y+4.5) {$\bullet$}; \node [below] at (\z+18,\y+4.3) {\small $(1,0)$};
\node (s3) at (\z+24,\y+4.5) {$\tl$}; \node [below] at (\z+24,\y+4.3) {\small $(0,0)$};
\node (s4) at (\z+30,\y+4.5) {$\tl$}; \node [below] at (\z+30,\y+4.3) {\small $(0,0)$};

 \path
(s0) edge [->,line width=1.5pt,color=black] (s1) 
(s1) edge [->,line width=1.5pt,color=black] (s2)
(s2) edge [->,line width=1.5pt,color=black] (s3)
(s3) edge [->,line width=1.5pt,color=black] (s4);

\node [above] at (\z+9,\y+4.6) {\small $2$};
\node [above] at (\z+15,\y+4.6) {\small $2$};
\node [above] at (\z+21,\y+4.6) {\small $1$};
\node [above] at (\z+27,\y+4.6) {\small $0$};

\node (s0) at (\z+6,\y+1.5) {$\bullet$}; \node [below] at (\z+6,\y+1.3) {\small $(0,1)$};
\node (s1) at (\z+12,\y+1.5) {$\tl$}; \node [below] at (\z+12,\y+1.3) {\small $(0,0)$};
\node (s2) at (\z+18,\y+1.5) {$\tl$}; \node [below] at (\z+18,\y+1.3) {\small $(0,0)$};
\node (s3) at (\z+24,\y+1.5) {$\bullet$}; \node [below] at (\z+24,\y+1.3) {\small $(0,1)$};
\node (s4) at (\z+30,\y+1.5) {$\bullet$}; \node [below] at (\z+30,\y+1.3) {\small $(0,1)$};

 \path
(s0) edge [->,line width=1.5pt,color=black] (s1) 
(s1) edge [->,line width=1.5pt,color=black] (s2)
(s2) edge [->,line width=1.5pt,color=black] (s3)
(s3) edge [->,line width=1.5pt,color=black] (s4);

\node [above] at (\z+9,\y+1.6) {\small $1$};
\node [above] at (\z+15,\y+1.6) {\small $0$};
\node [above] at (\z+21,\y+1.6) {\small $1$};
\node [above] at (\z+27,\y+1.6) {\small $2$};

\renewcommand{\y}{-12.5}

\node[right] at (-4,-7.75+2) {\small $\text{Span} =
\left( \begin{array}{cccc} (0,2] \\ (2,0]  \end{array} \right)$};

\node[right] at (-4,\y+2) {\small $H =
\left( \begin{array}{cccc} 2 &1 &0 &2 \\ 0 &2 &2 &2  \end{array} \right)$};

\node[] at (\z+6,\y-1.5) {$V'_0$};
\node[] at (\z+12,\y-1.5) {$V'_1$};
\node[] at (\z+18,\y-1.5) {$V'_2$};
\node[] at (\z+24,\y-1.5) {$V'_3$};
\node[] at (\z+30,\y-1.5) {$V'_4$};

\draw [dashed] (\z+4.3,\y-.5) -- (\z+7.7,\y-0.5) -- (\z+7.7,\y+5.5) -- (\z+4.3,\y+5.5) -- (\z+4.3,\y-.5); 
\draw [dashed]  (\z+10.3,\y-.5) -- (\z+13.7,\y-.5) -- (\z+13.7,\y+5.5) -- (\z+10.3,\y+5.5) -- (\z+10.3,\y-.5); 
\draw [dashed]  (\z+16.3,\y-.5) -- (\z+19.7,\y-.5) -- (\z+19.7,\y+5.5) -- (\z+16.3,\y+5.5) -- (\z+16.3,\y-.5); 
\draw [dashed]  (\z+22.3,\y-.5) -- (\z+25.7,\y-.5) -- (\z+25.7,\y+5.5) -- (\z+22.3,\y+5.5) -- (\z+22.3,\y-.5); 
\draw [dashed]  (\z+28.3,\y-.5) -- (\z+31.7,\y-.5) -- (\z+31.7,\y+5.5) -- (\z+28.3,\y+5.5) -- (\z+28.3,\y-.5); 

\node (s0) at (\z+6,\y+4.5) {$\tl$}; \node [below] at (\z+6,\y+4.3) {\small $(0,0)$};
\node (s1) at (\z+12,\y+4.5) {$\bullet$}; \node [below] at (\z+12,\y+4.3) {\small $(1,0)$};
\node (s2) at (\z+18,\y+4.5) {$\bullet$}; \node [below] at (\z+18,\y+4.3) {\small $(0,1)$};
\node (s3) at (\z+24,\y+4.5) {$\tl$}; \node [below] at (\z+24,\y+4.3) {\small $(0,0)$};
\node (s4) at (\z+30,\y+4.5) {$\tl$}; \node [below] at (\z+30,\y+4.3) {\small $(0,0)$};

 \path
(s0) edge [->,line width=1.5pt,color=black] (s1) 
(s1) edge [->,line width=1.5pt,color=black] (s2)
(s2) edge [->,line width=1.5pt,color=black] (s3)
(s3) edge [->,line width=1.5pt,color=black] (s4);

\node [above] at (\z+9,\y+4.6) {\small $2$};
\node [above] at (\z+15,\y+4.6) {\small $2$};
\node [above] at (\z+21,\y+4.6) {\small $1$};
\node [above] at (\z+27,\y+4.6) {\small $0$};

\node (s0) at (\z+6,\y+1.5) {$\bullet$}; \node [below] at (\z+6,\y+1.3) {\small $(1,0)$};
\node (s1) at (\z+12,\y+1.5) {$\tl$}; \node [below] at (\z+12,\y+1.3) {\small $(0,0)$};
\node (s2) at (\z+18,\y+1.5) {$\tl$}; \node [below] at (\z+18,\y+1.3) {\small $(0,0)$};
\node (s3) at (\z+24,\y+1.5) {$\bullet$}; \node [below] at (\z+24,\y+1.3) {\small $(0,2)$};
\node (s4) at (\z+30,\y+1.5) {$\bullet$}; \node [below] at (\z+30,\y+1.3) {\small $(1,0)$};

 \path
(s0) edge [->,line width=1.5pt,color=black] (s1) 
(s1) edge [->,line width=1.5pt,color=black] (s2)
(s2) edge [->,line width=1.5pt,color=black] (s3)
(s3) edge [->,line width=1.5pt,color=black] (s4);

\node [above] at (\z+9,\y+1.6) {\small $1$};
\node [above] at (\z+15,\y+1.6) {\small $0$};
\node [above] at (\z+21,\y+1.6) {\small $1$};
\node [above] at (\z+27,\y+1.6) {\small $2$};


\end{tikzpicture}
\end{center}
\caption{Two isomorphic trellis representations for the row space of the matrix $G$. 
Product construction (above) and BCJR construction (below)}\label{F:exF3}
\end{figure}

\bigskip

The product construction uses a basis of step functions for the vertex labeling. For a path $(\nu,c)$ in the basis of the label code every vertex along the path is labeled either $0$ or $e_\nu$, such that the set of labels $\{ e_\nu \}$ for the paths in the basis is  linearly independent. There are precisely two changes along the path: $v_i = 0$ and $v_{i+1} = e_\nu$ for a unique $i$, and $v_j = e_\nu$ and $v_{j+1} = 0$ for a unique $j$ (an allzero path will not be part of the basis and the special case of a never zero path is needed only for disconnected trellises). The transitions are determined by choosing a span $(i,j]$ for the path such that $c_i, c_j \neq 0$ and $c_k = 0$ whenever $k$ is \emph{not in the span}.\footnote{$k \not \in (i,j]$ if the ordered triple $i,j,k$ is even, i.e., $i < j < k$, $k < i < j$ or $j < k < i$.} The span of a path is in general not unique. 
A Koetter-Vardy trellis is the special case where the selected generators have a span $(i,j]$ that is minimal (among all spans of the form $(i',j]$ or $(i,j']$).  
In general a trellis has precisely $n$ such minimal spans $(i,j]$, each with a unique $i \in \{0,1,\ldots,n-1\}$ and a unique $j \in \{0,1,\ldots,n-1\}$. 
The product construction was introduced for conventional trellises in \cite{KS95}. 

\bigskip

The vertex labeling in the tail-biting version of the BCJR construction depends only on the labeling of the vertices in $V_0$. The labeling of the other vertices along a path follows the rule $v_{i+1} = v_i + c_i h_i$ of the classical BCJR construction, for parity-check vectors $h_i$. Note that $v_n = v_0 + \sum c_i h_i$. To guarantee that the trellis represents the linear space $C$, the $h_i$ are chosen such that $\sum c_i h_i =0$ if and only if $c \in C$. A span based BCJR trellis is the special case where the label $v_0$ in a path is decided by the rule $c_k = 0$ whenever $k$ is \emph{not in the span} together with the rule $v_{i+1} = v_i + c_i h_i$. The tail-biting versions of the BCJR construction were introduced by Nori and Shankar. Their formulation emphasizes coset decompositions of the row space. Coset decompositions were introduced for conventional trellises in \cite{F88}. 

\bigskip

The two constructions give two ways to describe minimal trellises.
\[
\{ \text{minimal trellises} \} \subset \{ \text{KV-trellises} \} \subset \{ \text{span based BCJR trellises} \} 
\]
The inclusions hold for trellises up to isomorphism. Thus a KV-trellis can be constructed as a span based BCJR trellis up to isomorphism, as in Figure \ref{F:exF3}. The first inclusion is proved in \cite{KV03}, \cite{GW11a}, and the second in \cite{NS06}, \cite{GW11a}. The definitions in \cite{GW11a}, which we follow for this paper, differ slightly from those in \cite{KV03} and \cite{NS06}. 

\bigskip

Koetter and Vardy \cite{KV03} define a characteristic matrix $X$ for a given matrix $G$ as a square matrix with the same row space as $G$ such that the row spans start and end in different positions. The matrices $G$ and $H$ in Figure~\ref{F:exF3} have characteristic matrices $X$ and $Y$ (the rows are ordered such that spans of rows in $X$ start on the diagonal and spans of rows in $Y$ end on the diagonal).  
\begin{equation} \label{E:exF3}
X = \left[ \begin{array}{cccc}
2 &2 &1 &0 \\
0 &1 &1 &1 \\
1 &0 &1 &2 \\
2 &1 &0 &2 
\end{array} \right] \quad \text{and} \quad 
Y = \left[ \begin{array}{cccc}
1 &0 &1 &2 \\
1 &2 &0 &1 \\
1 &1 &2 &0 \\
0 &1 &1 &1 
\end{array} \right]
\end{equation}
In general, the characteristic matrix is not unique (even after ordering and normalizing). For any minimal trellis, there exists a characteristic matrix such that the trellis can be constructed as a product trellis from rows in the characteristic matrix (\cite{KV03}, \cite{GW11a}). The product trellis in Figure~\ref{F:exF3} is a minimal trellis and was constructed using rows $1$ and $3$, and their spans $(0,2]$ and $(2,0]$, in the matrix $X$. While characteristic matrices are not unique, there exists a bijection between the characteristic matrices for a pair of matrices $G$ and $H$ with maximal orthogonal row spaces such that trellises constructed from matching characteristic matrices $X$ and $Y$ are in duality \cite{GW11b}.

\subsection*{Main results}

The main results are presented in Sections \ref{S:char} and \ref{S:dual}. There we define a characteristic matrix in a novel way and from the definition we easily obtain its main properties including a straighforward description of the relation between a characteristic matrix and its dual. An important part of our definition is the notion of a reduced minimal span form, which is developed in Sections \ref{S:mins} and \ref{S:reds}. 

\bigskip      

Minimal span matrices (the subject of Sections \ref{S:mins} and \ref{S:reds}) are charactrerized by the left-right property: no two rows start or end in the same positions. In other words, only row exchanges are needed to bring the matrix in either left or right echelon form. Their connection to trellises is two-fold. The product construction applied to the rows of a matrix yield a minimal trellis if and only if the matrix is in minimal span form. Secondly, the vertex count of a minimal conventional trellis can be read off easily from the shape of a matrix in minimal span form. Throughout the literature, in many examples and trellis descriptions, it is noted that the minimal span form is in general not unique and it appears that some arbitrary choices need to be made to settle for a particular minimal span form. In this paper we bring a matrix in reduced minimal span form in the same way that back substitution brings a matrix from echelon form into reduced echelon form. This reduction is equivalent to reducing an arbitrary $LPU$ decomposition of a nonsingular matrix 
to a unique reduced $LPU$ decomposition. In this form our proposed reduction is the same as the one used for reduced Bruhat decompositions in \cite{YK87}.

\bigskip

In Section \ref{S:char} we adopt a definition of characteristic matrices that refers to a matrix in minimal span form. The notion of reduced minimal span form thus allows us to define a unique reduced characteristic matrix. We use two different forms for a characteristic matrix, a left-ordered form (controller form) and a right-ordered form (observer form). Let $G$ and $H$ be the pair of orthogonal matrices from Figure~\ref{F:exF3}, with characteristic matrices given by (\ref{E:exF3}). For the controller form of the characteristic matrix $X$ for $G$, we write $X$ as the sum of an upper triangular matrix and a lower triangular matrix. 
\[ 
X = X_0 + X_1
= \left[ \begin{array}{cccc}
2 &2 &1 &0 \\
\dt &1 &1 &1 \\
\dt &\dt &1 &2 \\
\dt &\dt &\dt &2 
\end{array} \right] +
\left[ \begin{array}{cccc}
\dt &\dt &\dt &\dt \\
0 &\dt &\dt &\dt \\
1 &0 &\dt &\dt \\
2 &1 &0 &\dt 
\end{array} \right]
\]
The matrices $G$ and $X$ share the same row space. Since $X H^T =0$, $(X_0 | X_1)(H | H)^T = 0$. In general we derive $X_0$ and $X_1$ from a reduced minimal span form for the null space of $(H|H)$.
The dual characteristic matrix $Y$ is decomposed as follows. 
\[
Y = Y_1 + Y_0 = \left[ \begin{array}{cccc}
\dt &0 &1 &2 \\
\dt &\dt &0 &1 \\
\dt &\dt &\dt &0 \\
\dt &\dt &\dt &\dt 
\end{array} \right] +
\left[ \begin{array}{cccc}
1 &\dt &\dt &\dt \\
1 &2 &\dt &\dt \\
1 &1 &2 &\dt \\
0 &1 &1 &1 
\end{array} \right]
\]
It satisfies $Y G^T =0$ and thus $(Y_1 | Y_0) (G | G)^T =0$. Similar to the above, we derive $Y_1$ and $Y_0$ using a reduced minimal span form for the null space of $(G|G)$. 
In Section \ref{S:dual}, we express duality for characteristic matrices as
\[
\left( \begin{array}{cc} X_0 &X_1 \\ 0 &X_0 \end{array} \right) \left( \begin{array}{cc} Y_0 &0 \\ Y_1 &Y_0 \end{array} \right)^T = \left( \begin{array}{rr} -I &I \\ 0 &-I \end{array} \right).
\]
In Theorem \ref{T:trdual} we prove, avoiding some of the technical details of earlier proofs, that KV-trellises defined with $X$ and with $Y$ are in duality. Furthermore, in Theorem \ref{T:col} we show that two characteristic matrices are in duality if and only if they have orthogonal column spaces. And in Theorem \ref{T:trans} that the transpose of a characteristic matrix is again a characteristic matrix if and only if the characteristic matrix is reduced.  

Originally Koetter and Vardy proposed that a duality between characteristic matrices holds for the lexicographically first characteristic matrices. In \cite{GW11a} examples were found where this was not the case. In Theorem \ref{T:conj} we adjust the original conjecture and prove that reduced characteristic matrices are in duality. Reduced characteristic matrices are lexicographically first when read out from right to left (for $X$) or left to right (for $Y$). After taking into account the different directions of the time axis in a trellis and its dual the original conjecture appears correct.   

\bigskip

To describe the KV-trellises both as a product trellis and as a BCJR trellis, we use the column space of $Y$as labels for trellises constructed with rows in $X$. 
\[
Y^T = \left[ \begin{array}{cccc}
1 &1 &1 &0 \\
0 &2 &1 &1 \\
1 &0 &2 &1 \\
2 &1 &0 &1
\end{array} \right]
\]
This leads to label codes of the form 
\[
S(Y^T|X) = \left[ \begin{array}{c|c|c|c|c|c|c|c|c}
0~0~0~0   &2   &1~1~1~0   &2   &1~0~2~1   &1   &0~0~0~0   &0   &0~0~0~0  \\
0~0~0~0   &0   &0~0~0~0   &1   &0~1~2~2   &1   &2~1~0~1   &1   &0~0~0~0  \\
1~1~1~0   &1   &0~0~0~0   &0   &0~0~0~0   &1   &2~0~1~2   &2   &1~1~1~0  \\
2~1~0~1   &2   &0~2~1~1   &1   &0~0~0~0   &0   &0~0~0~0   &2   &2~1~0~1    
\end{array} \right]
\]
We prove in Section \ref{S:disp} that the following properties hold in general. 
The single columns in the matrix lie in the null space of the square matrices (Lemma \ref{T:col}).
Any two consecutive square matrices differ by a rank-one matrix (Lemma \ref{L:disp2}).
And each of the square matrices contains a basis in echelon form for the row space $Y^T$ (Theorem \ref{T:disp}).
Moreover, by Theorem \ref{T:trdual} the square matrices in the dual label code $S(X^T|Y)$ relate to those of  $S(Y^T|X)$ by transposition (up to a sign change and only off the diagonal).

\section{Minimal span form of a matrix} \label{S:mins}

 The span (or support interval) of a nonzero vector $c = (c_0, c_1, \ldots, c_{n-1})$ is the largest interval $[i_0,i_1] \subset [0,n-1]$ with $c_{i_0} c_{i_1} \neq 0$. We say the vector starts in $i_0$ and ends in $i_1$. The size $i_1 - i_0$ of the span is called the spanlength of the vector. The spanlength of a matrix is the sum of the spanlengths of its rows. A matrix is in minimal span form if its spanlength is minimal among all row equivalent matrices. An easy way to decide if a matrix is in minimal span form is via the \emph{left-right property}.  

\begin{lemma}[\cite{KS95}, \cite{McE96}] \label{L:sp1}
The following are equivalent.
\begin{itemize}
\item[(1)] A matrix is in minimal span form.
\item[(2)] No two distinct rows in the matrix start in the same position or end in the same position. 
\end{itemize}
\end{lemma}

\begin{proof} The proof is short and elementary. If (2) fails for two rows in a matrix then we can reduce the spanlength of one of the two rows and (1) fails. If (2) holds for a matrix then the rows of the matrix are in row echelon form after reordering and the starting positions agree with the column indices for the left pivots. By the same argument the ending positions agree with the column indices for the right pivots. Thus the spanlength is the same for all matrices with (2). By the first part of the proof this spanlength is that of a matrix in minimal span form.
\end{proof}

Let $I_0$ be the collection of left pivot positions in a matrix and let $I_1$ be the collection of right pivots positions. From the proof of the lemma we see that the spanlength for a matrix in minimal span form is given by $\sum_{I_1} i_1 - \sum_{I_0} i_0.$

\begin{lemma}[\cite{KS95}, \cite{McE96}] \label{L:sp2}
Any two minimal span forms for a matrix have the same collection of spans.
\end{lemma}
 \begin{proof} The proof of the previous lemma shows that two minimal span forms share the same set $I_0$ of starting indices.
For a given starting index $i_0 \in I_0$ both minimal span forms contain a row that starts in $i_0$. Let the two rows have spans $[i_0,i_1]$ and $[i_0,i'_1]$ and assume that $i_1 < i'_1$. Then the row with span $[i_0,i'_1]$ can be reduced to a row with smaller span $[i'_0,i'_1]$. The reduced row is still independent of the other rows since $i'_1$ has not changed and this contradicts that the matrix was in minimal span form. Thus $i_1 = i'_1$ as required. 
\end{proof}

The left-right property in Lemma \ref{L:sp1} expresses that after ordering the rows of a matrix in minimal span form we may assume that the minor in the positions $I_0$ (the leading nonsingular full minor in the matrix) is upper triangular and that the minor in the positions $I_1$ (the trailing nonsingular full minor in the matrix) is lower triangular up to row permutations. Denoting these two minors by $U$ and $PL$, respectively, the spans are determined by the permutation matrix $P$. 
Lemma \ref{L:sp2} expresses that $P=P'$ for any two row equivalent matrices $[U | PL] \sim [U' | P'L']$. We relate this last property to the $LPU$-decomposition of an invertible matrix $A$.

\begin{lemma} ($LPU$-decomposition) \label{L:lpu}
Let $A$ be an invertible matrix. Then there exists a decomposition $A = LPU$ with $L$ lower triangular, $P$ a permutation matrix and $U$ upper triangular. The decomposition is in general not unique but the matrix $P$ is uniquely determined by the matrix $A$.
\end{lemma}

\begin{proof} The result is well known.\footnote{But should not be confused with the different decomposition $PA=LU$.}
To see how the uniqueness claim
\begin{equation} \label{eq:2}
L P U = L' P' U'  ~\Rightarrow~ P=P'
\end{equation} 
follows from the previous lemmas, note that
\[
L P U = L' P' U'   ~\Leftrightarrow~  [ \; U | P^{-1}L^{-1}\; ] \sim [ \;U' | P'^{-1} L'^{-1}\; ] 
\]
and apply Lemma \ref{L:sp2},
\begin{equation} \label{eq:1}
 [ \; U | P^{-1}L^{-1} \; ] \sim  [ \;U' | P'^{-1} L'^{-1}\; ] ~\Rightarrow~ P=P'.
\end{equation} 
\end{proof}

\subsection{Bruhat decomposition} \label{S:bruh}

The uniqueness claim (\ref{eq:2}) is an instance of a modified Bruhat decomposition for the general linear group. In general, a Bruhat decomposition of a (connected, reductive) linear algebraic group is of the form $G = \cup_{w \in W} BwB$, for a Borel subgroup $B$ and a Weyl group $W$. 
The double coset decomposition $\cup_{w \in W} BwB$ partitions the group $G$ and associates to each group element $g$ a unique element $w \in W$ (e.g . \cite[Theorem 10.2.7 (``Bruhat's lemma'')]{Sp09}). For the general linear group, the Bruhat decomposition uses for $B$ the group of invertible upper triangular matrices and for $W$ the group of permutation matrices. 

\bigskip

The standard Bruhat decomposition assigns to a matrix $A$ a unique permutation matrix $U_1^{-1} A U_2^{-1} = P$. In general it is possible to reduce an invertible matrix $A$ to a unique permutation matrix $P$ using row operations and column operations of any fixed triangular type. The outcome $P$ depends on the chosen types but not on the sequence of reduction operations. To illustrate that the outcome depends on the chosen types we reduce the same binary matrix $A$ in two different ways to a permutation matrix, via a standard Bruhat decomposition and via a $LPU$-decomposition.
\begin{align*}
U_1^{-1} A &= \left( \begin{array}{rrr} 1 &1 &0 \\ 0 &1 &0 \\ 0 &0 &1 \end{array} \right) 
\left( \begin{array}{rrr} 1 &1 &0 \\ 1 &1 &1 \\ 0 &1 &1 \end{array} \right) =
\left( \begin{array}{rrr} 0 &0 &1 \\ 1 &1 &1 \\ 0 &1 &1 \end{array} \right) = 
\left( \begin{array}{rrr} 0 &0 &1 \\ 1 &0 &0 \\ 0 &1 &0 \end{array} \right) U_2. \\ 
L^{-1} A &= \left( \begin{array}{rrr} 1 &0 &0 \\ 1 &1 &0 \\ 0 &0 &1 \end{array} \right) 
\left( \begin{array}{rrr} 1 &1 &0 \\ 1 &1 &1 \\ 0 &1 &1 \end{array} \right) =
\left( \begin{array}{rrr} 1 &1 &0 \\ 0 &0 &1 \\ 0 &1 &1 \end{array} \right) = 
\left( \begin{array}{rrr} 1 &0 &0 \\ 0 &0 &1 \\ 0 &1 &0 \end{array} \right) U. 
\end{align*}

\bigskip

To interpret the matrix $P$ assigned to $A$ by the $LPU$ decomposition, define a leading submatrix of a matrix as a submatrix positioned in the north-west corner of the matrix. The following is clear and stated without proof.

\begin{lemma}
Let $A$ be any matrix. For any lower triangular invertible matrix $L$ and any upper triangular invertible matrix $U$, leading submatrices in $A$ and $LAU$ of the same size are of the same rank.  
\end{lemma}

It follows that the matrix $P = L^{-1}AU^{-1}$, that is assigned to $A$ by an $LPU$-decomposition, describes the rank of any leading submatrix in $A$ by counting the number of nonzero entries in the corresponding leading submatrix of $P$. In particular the matrix $P$ is unique, which confirms in a different way the claim (\ref{eq:2}). 
Similarly, for each of the four corners of the matrix, we can reduce $A$ to a permutation matrix $P$ that gives the rank of submatrices positioned in that corner:
\[
\begin{array}{lcl}
 P_{NW} = L^{-1}AU^{-1}, &~~~~ &P_{NE} = L^{-1}AL^{-1} \\
 P_{SW} = U^{-1}AU^{-1}, &~~~~ &P_{SE} = U^{-1}AL^{-1}.
\end{array}
\]

\begin{lemma} \label{L:T}
Let $PT= T'P'$, for invertible triangular matrices $T$ and $T'$ and for permutation matrices $P$ and $P'$. Then $P=P'$.
\end{lemma}

\begin{proof}
Submatrices of $P$ and $P'=T'^{-1}PT$ of the same size in one of the four corners (which corner depends on the type of $T$ and $T'$, upper or lower triangular) are of the same rank and thus $P=P'$. 
\end{proof}

\section{Reduced minimal span form of a matrix} \label{S:reds}

Let $G$ be a given matrix.
As in the previous section, let $I_0$ be the collection of left pivot positions and let $I_1$ be the collection of right pivot positions. As pointed out after Lemma \ref{L:sp2} the rows in a minimal span form can be ordered such that the matrix is in echelon form on the left.
\begin{equation} \label{msf}
G|_{I_0} = U ~~~\text{and}~~~ G|_{I_1} = PL.
\end{equation}
We will call a minimal span form of this type left-ordered.  
An example is the binary matrix
\begin{center}
    \begin{tikzpicture}
        [scale=.5,auto=center,every node/.style={minimum size=0em}]
        \node at (1,11.6) {$G ~= $};      
        \matrix [matrix of math nodes,left delimiter=(,right delimiter=)] (m) at (6,11.6)
        {
             1 &1 &1 &0 &0 \\ 0 &1 &1 &0 &1 \\ 0 &0 &1 &1 &0 \\
        };  
        \draw[color=black] (m-1-1.north west) -- (m-1-3.north east) -- (m-3-3.south east) -- (m-3-1.south west) -- (m-1-1.north west);
        \node at (10.5,11.6) {$ = $};
        \matrix [matrix of math nodes,left delimiter=(,right delimiter=)] (m) at (15,11.6)
        {
             1 &1 &1 &0 &0 \\ 0 &1 &1 &0 &1 \\ 0 &0 &1 &1 &0 \\
        };  
        \draw[color=black] (m-1-3.north west) -- (m-1-5.north east) -- (m-3-5.south east) -- (m-3-3.south west) -- (m-1-3.north west);
    \end{tikzpicture}
\end{center}

Among all minimal span forms (\ref{msf}) we select a unique reduced minimal span form.

\begin{lemma} \label{L:3}
There exists a minimal span form of the form (\ref{msf}) with the additional property that the trailing nonzero element in each row is the first nonzero element in its column.
\end{lemma}

\begin{proof}
Let $(i,\pi(i))$, $1 \leq i \leq n$, be the nonzero entries in $P$, for a permutation $\pi$ of $\{1,2,\ldots,n\}$. Thus $\pi(i)$ marks the last column in row $i$ of $PL$ with a nonzero entry. The nonzero entries in column $\pi(i)$ with
row index $i' < i$ have $\pi(i') > \pi(i)$. The pair $i,i'$ corresponds to an inversion of $\pi$. Using a sequence of elementary row operations in upper triangular form we can clear all nonzero entries in column $i$ with row index $i' < i$. To assure that columns that have been cleared are not affected by clearing other columns, columns should be cleared in reverse order, from last to first column.
\end{proof}

A matrix $PL$ has the property in the lemma if and only if $PL = L'P'$ for a lower triangular matrix $L'$ and a permutation matrix $P'$. By Lemma \ref{L:T} we may assume that $P'=P$.
 
\begin{definition} \label{D:rmsf}
A minimal span form is called left-ordered and right-reduced if 
\begin{equation} \label{rmsf}
G|_{I_0} = U ~~~\text{and}~~~ G|_{I_1} = PL = L'P.
\end{equation}
It is upper triangular in the leading pivot columns and lower triangular up to a row permutation in the trailing pivot columns. Moreover it is lower triangular up to a column permutation in the trailing pivots. 
Less formal, a minimal span form is reduced if it is upper triangular on the leading pivots (and thus zero in the positions west and south of the leading pivots) and zero in the positions east and north of the trailing pivots.
\end{definition}
 
\begin{proposition}
Every matrix has a unique reduced minimal span form.
\end{proposition}

\begin{proof}
Lemma \ref{L:3} proves the existence of a minimal span form $G$ that satisfies (\ref{rmsf}). Lemma \ref{L:sp2} shows that the matrix $P$ is the same for any two minimal span forms. Thus, the matrix $U^{-1} L'$ is the same for any two reduced minimal span forms. It follows that $U$ and $L'$ are unique (up to a constant diagonal matrix).
\end{proof}

\subsection{Dual minimal span forms}

Definition \ref{D:rmsf} characterizes the reduced minimal span form of a matrix via the properties left-ordered and right-reduced. This follows the convention for a reduced row echelon form where the echelon form appears on the left. In both cases forward elimination brings the matrix in echelon form on the left. The back substitution phase of Gaussian elimination clears columns on the left to obtain the reduced row echelon form. For the reduced minimal span form the second phase of clearing columns is carried out on the opposite right side of the matrix. 
  
As in \cite[Section IV]{McE96}, the standard row reduced echelon form might be called a ``left'' reduced row echelon form. If pivot columns are selected from the right, one obtains what might be called a ``right'' reduced row echelon form. We denote a left reduced echelon form for $G$ by $G_0$ and a right reduced echelon form by $G_1$.
\[
G_0 |_{I_0} = I, ~~~~G_1 |_{I_1} = I.
\]
For a pair of matrices $G$ and $H$ with maximal orthogonal row spaces, the matrix $H$ has leading pivots $J_0 = \{ 0,1,\ldots, n-1 \} \backslash I_1$ and trailing pivots $J_1 = \{ 0,1,\ldots, n-1 \} \backslash I_0$. The echelon forms $G_0$ and $H_1$ are in duality via
\[
G_0 |_{I_0} = I, ~~~~G_0 |_{J_1} = A, ~~~~ H_1 |_{I_0} = -A^T, ~~~~H_1 |_{J_1} = I. 
\]
To obtain a duality for reduced minimal span forms we add to Definition \ref{D:rmsf} a matching reduced minimal span form for the matrix $H$.

\begin{definition} \label{D:rmsf2}
A matrix $G$ is in left-ordered right-reduced minimal span form if
\[
G|_{I_0} = U ~~~\text{and}~~~ G|_{I_1} = PL = L'P.       
\]
A matrix $H$ is in right-ordered left-reduced minimal span form if
\[
H|_{I_0} = PU = U'P ~~~\text{and}~~~ H|_{J_1} = L.          
\]
The minimal span form for $G$ is normalized by choosing $L$ with $1$s on the diagonal. The minimal span form for $H$ is normalized by choosing $U$ with $1$s on the diagonal.
\end{definition}
  
We write $G_{01}$ for a reduced minimal span form of the first type and $H_{10}$ for a reduced minimal span form of the second type. We will rarely use the full long form terminology of the definition. Instead we write $G \perp H$ for a pair of matrices with maximal orthogonal row spaces and call the pair reduced if $G=G_{01}$ and $H=H_{10}$. This is illustrated for the pair $G_{01} \perp H_{10}$ below.

\begin{center}
    \begin{tikzpicture}
        [scale=.5,auto=center,every node/.style={minimum size=0em}]
        \node [left] at (-4,11.6) {$G_0 = $};      
        \matrix [matrix of math nodes,left delimiter=(,right delimiter=),right] (m) at (-3.5,11.6)
        {
             1 &0 &0 &0 &1 \\ 0 &1 &0 &1 &1 \\ 0 &0 &1 &1 &0 \\
        };  
        \draw[color=black] (m-1-1.north west) -- (m-1-3.north east) -- (m-3-3.south east) -- (m-3-1.south west) -- (m-1-1.north west);
        \node [left] at (9,11.6) {$G_{01} ~= $};      
        \matrix [matrix of math nodes,left delimiter=(,right delimiter=),right] (m) at (9.5,11.6)
        {
             1 &1 &1 &0 &0 \\ 0 &1 &1 &0 &1 \\ 0 &0 &1 &1 &0 \\
        };  
        \draw[color=black] (m-1-1.north west) -- (m-1-3.north east) -- (m-3-3.south east) -- (m-3-1.south west) -- (m-1-1.north west);
        \node [left] at (17,11.6) {$ = $};
        \matrix [matrix of math nodes,left delimiter=(,right delimiter=),right] (m) at (17.5,11.6)
        {
             1 &1 &1 &0 &0 \\ 0 &1 &1 &0 &1 \\ 0 &0 &1 &1 &0 \\
        };  
        \draw[color=black] (m-1-3.north west) -- (m-1-5.north east) -- (m-3-5.south east) -- (m-3-3.south west) -- (m-1-3.north west);
        \node [left] at (-4,8) {$H_1 = $};      
        \matrix [matrix of math nodes,left delimiter=(,right delimiter=),right] (m) at (-3.5,8)
        {
              0 &1 &1 &1 &0  \\ 1 &1 &0 &0 &1  \\
        };  
        \draw[color=black] (m-1-4.north west) -- (m-1-5.north east) -- (m-2-5.south east) -- (m-2-4.south west) -- (m-1-4.north west); 
        \node [left] at (9,8) {$H_{10} ~= $};      
        \matrix [matrix of math nodes,left delimiter=(,right delimiter=),right] (m) at (9.5,8)
        {
              0 &1 &1 &1 &0  \\ 1 &0 &1 &1 &1   \\
        };  
        \draw[color=black] (m-1-4.north west) -- (m-1-5.north east) -- (m-2-5.south east) -- (m-2-4.south west) -- (m-1-4.north west);
        \node [left] at (17,8) {$ = $};      
        \matrix [matrix of math nodes,left delimiter=(,right delimiter=),right] (m) at (17.5,8)
        {
             0 &1 &1 &1 &0  \\ 1 &0 &1 &1 &1   \\
        };  
        \draw[color=black] (m-1-1.north west) -- (m-1-2.north east) -- (m-2-2.south east) -- (m-2-1.south west) -- (m-1-1.north west);
    \end{tikzpicture}
\end{center}

The property that characterizes a matrix in minimal span form is the LR-property (property (2) in Lemma \ref{L:sp1}). Similarly, we may call the characterizing property of a reduced minimal span form the SW-NE property (after ordering the rows on the left or on the right, entries south and west of a leading pivot are zero and entries north and east of a trailing pivot are zero). This property will have a clear interpretation in the duality of characteristic matrices in Section \ref{S:char}. 

\subsection{Reduced Bruhat decompositions}

Given the connection between minimal span forms and Bruhat decompositions that was pointed out in Section \ref{S:bruh}, we want to interpret the reduced minimal span form of a matrix in terms of Bruhat decompositions. This leads us to the notion of a reduced Bruhat decompositon defined by Kolotrilina and Yeremin in \cite{YK87}. About a Bruhat decomposition $A = V \pi_A U$ ($A$ nonsingular, $U$ and $V$ upper triangular) they write [loc.sit. page 422]: ``The triangular factors $V$ and $U$ $\ldots$ are not in general uniquely determined. To remove this uncertainty, one can consider the so-called reduced Bruhat decomposition.''

\begin{definition}[\cite{YK87}] \label{D:KY} The decompositon $A = V \pi_A U$ will be referred to as the reduced-on-the-right or on-the-left Bruhat decompositon if
\[
(\pi_A) U (\pi_A)^T  ~~~~\text{or}~~~~ (\pi_A) U (\pi_A)^T,
\]
respectively, is lower triangular.
\end{definition}

Further on [loc.sit. page 423]: ``It is easy to show that all factors of the reduced Bruhat decompositon are uniquely determined.''
The defintion is motivated by applications in numerical analysis [loc.sit page 421]: ``the Bruhat decomposition $\ldots$ allows the formulation of a global analytical criterion for optimizing sparsity patterns of triangular factors $\ldots$''. Clearly the reductions applied in Definition \ref{D:rmsf2} agree with those of Definition \ref{D:KY} when the latter are rewritten for a modified Bruhat decomposition. 

\bigskip  

The nonuniqueness of matrices $L$ and $U$ in the decomposition $A=LPU$ is due to inversions in the permutation matrix $P$.
This is illustrated by
\[
A = \left( \begin{array}{cc} 0&1 \\ 1 &0 \end{array} \right) =
\left( \begin{array}{cc} 1&0 \\ -a &1 \end{array} \right) 
\left( \begin{array}{cc} 0&1 \\ 1 &0 \end{array} \right)
\left( \begin{array}{cc} 1&a \\ 0 &1 \end{array} \right) = LPU.
\]
To the one-parameter family of LPU-decompositions corresponds the one-parameter family of minimal span forms 
\[
 ( A | I ) = \left( \begin{array}{cc|cc} 0&1 &1 &0 \\ 1 &0 &0 &1 \end{array} \right) \sim
\left( \begin{array}{cc|cc} 1&a &a &1 \\ 0&1 &1&0 \end{array} \right) = ( U | P^{-1} L^{-1} )
\]
Clearly, the choices for $L$ and $U$ that minimize the number of operations in the assignment $A \mapsto P = L^{-1} A U^{-1}$ (in other words that optimize the sparsity of $L^{-1}$ and $U^{-1}$) are the choices $L=U=I$ and $a=0$. Which is our choice for the reduced minimal span form of the matrix $(A|I)$.

\section{Characteristic matrix} \label{S:char}

Koetter and Vardy \cite{KV03} define a redundant set of generators, called characteristic generators, for the row space of a matrix. From the redundant set it is straightforward to recover not only a minimal span form for the original matrix but also for the cyclic permutations of the original matrix. 
In \cite{KV03}, the set of characteristic generators is built using an efficient procedure. The procedure checks the minimal span forms for cyclicly permuted versions of the original matrix one cyclic shift at a time and each time decides if a new generator should be added to the list. The final list is shown to be of size $n$ and has the property that the spans for the generators start and end in distinct positions. The main result of \cite{KV03} is that any reduced linear tail-biting trellis can be built with the product construction from a subset of characteristic generators. 
For this result, the existence of characteristic generators and the fact that they can be computed efficiently when needed is sufficient. 
In this section we describe characteristic generators in a different way and establish new properties.
The properties lead to short proofs for duality of Koetter-Vardy tail-biting trellises. 

\begin{definition}
A row vector $(c_0,c_1,\ldots,c_{n-1})$ has span $(i,j]$ if $c_i, c_j \neq 0$ and if $c_k = 0$ whenever $k$ is \emph{not in the span}. Here $k \not \in (i,j]$ if and only if the ordered triple $i,j,k$ is even, i.e., $i < j < k$, $k < i < j$ or $j < k < i$. The span of a row vector is in general not unique. 
\end{definition}

The definition generalizes the definition of conventional span (or support interval) used in Section~\ref{S:mins}. The length of a conventional span $(i,j]$ is $j-i$ and it is $n-(j-i)$ for its complement $(j,i]$.

\begin{definition} \label{D:char}
Let $G$ be a matrix, $C = \row G$ its row space, and $n$ the number of columns.
A \emph{set of characteristic generators} for $C$ is a subset of $n$ vectors such that the spans of the vectors start and end in distinct positions and such that the sum of the spanlengths of the vectors is minimal. A square matrix is called a characteristic matrix for $G$ if its rows form a set of characteristic generators for the row space of $G$. 
\end{definition}

We call $X$ a left-ordered characteristic matrix if rows are ordered such that each row has span starting on the diagonal, and in that case write $X = X_0 + X_1$, for an upper triangular matrix $X_0$ and a strictly lower triangular matrix $X_1$. We call $Y$ a right-ordered characteristic matrix if rows are ordered such that each row has span ending on the diagonal, and in that case write $Y = Y_1 + Y_0$, for a lower triangular matrix $Y_0$ and a strictly upper triangular matrix $Y_1$. The unwrapped version $(X_0|X_1)$ of the matrix $X$ is in left-ordered minimal span form and the unwraped version $(Y_1|Y_0)$ of the matrix $Y$ is in right-ordered minimal span form. 
Circulant matrices are a special case of characteristic matrices. The unwrapping of two circulant matrices is shown below, 
for polynomials $c(x)=c_0+c_1 x + c_2 x^2$ and $d(x) = d_0 + d_1 x + d_2 x^2$ such that $c(x)d(x) = x^4-1$. 
\[
X = \left( \begin{array}{ccccc}
c_0  &c_1 &c_2 &0 \\
0 &c_0  &c_1 &c_2 \\
c_2 &0 &c_0  &c_1  \\ 
c_1 &c_2 &0 &c_0 
\end{array} \right)
~~~~(X_0 | X_1)  = 
\left( \begin{array}{cccc|cccc}
c_0  &c_1 &c_2 &0 &0 &0  &0 &0 \\
0 &c_0  &c_1 &c_2 &0 &0 &0 &0 \\
0 &0 &c_0  &c_1 &c_2 &0  &0 &0 \\
0 &0 &0 &c_0 &c_1 &c_2  &0 &0 
\end{array} \right)  
\]
\[
Y = \left( \begin{array}{ccccc}
d_0 &0 &d_2 &d_1  \\
d_1 &d_0 &0 &d_2 \\
d_2 &d_1 &d_0 &0 \\ 
0 &d_2 &d_1 &d_0 
\end{array} \right)
~~~~(Y_1|Y_0) =
\left( \begin{array}{cccc|cccc}
0 &0 &d_2 &d_1 &d_0 &0 &0 &0 \\
0 &0 &0 &d_2 &d_1 &d_0 &0 &0 \\
0 &0 &0 &0 &d_2 &d_1 &d_0 &0 \\ 
0 &0 &0 &0 &0 &d_2 &d_1 &d_0 
\end{array} \right)
\]
The spanlength of a vector before unwrapping agrees with the conventional spanlength after unwrapping. 
Let $G \perp H$ be a pair of matrices with maximal orthogonal row spaces, let $X$ be a left-ordered characteristic matrix for $G$ and $Y$ a right-ordered characteristic matrix for $H$. 
Unwrapping of the rows in $X$ yields a matrix $(X_0 | X_1) \perp (H|H)$. Since $X_0$ is invertible (its rows start in distinct positions), the two matrices
\begin{equation} \label{E:X}
\left( \begin{array}{c|c} X_0 &X_{1} \\ \hline \noalign{\smallskip} 0 &G \end{array}\right)  \perp (H|H)
\end{equation}
have maximal orthogonal row spaces. Similarly, the matrix $Y_0$ in the unwrapped version $(Y_1|Y_0)$ of $Y$ is invertible and the matrices
\begin{equation} \label{E:Y}
(G|G) \perp \left( \begin{array}{c|c} H &0 \\ \hline \noalign{\smallskip} Y_1 &Y_0 \end{array}\right)
\end{equation}
have maximal orthogonal row spaces. 

\begin{lemma} \label{L:gen} 
Let $G \perp H$ be a pair of matrices with maximal orthogonal row spaces. 
{\begin{itemize}
\item[(X)] A matrix $X=X_0+X_1$ is a left-ordered characteristic matrix for $G$ if and only if the rows of $(X_0|X_1)$ together with those of $(0\,|\,G)$ form a matrix in left-ordered minimal span form orthogonal to $(H|H)$. 
\item[(Y)] A matrix $Y=Y_1+Y_0$ is a right-ordered characteristic matrix for $H$ if and only if the rows of $(H\,|\,0)$ together with those of $(Y_1|Y_0)$ form a matrix in right-ordered minimal span form orthogonal to $(G|G)$.
\end{itemize}}
\end{lemma}

\begin{proof}
Minimizing the span of $x$ is the same as minimizing the conventional span of $(x_0|x_1)$. The minimum for the latter is attained for a matrix in minimal span form.  
\end{proof}

For the special choices of a reduced minimal span form in the lemma, we obtain what we will call a reduced characteristic matrix.

\begin{definition} \label{D:red} Let $G \perp H$ as in the lemma.
\begin{itemize}
\item[(X)] The matrix $X=X_0+X_1$ is a left-ordered, right-reduced characteristic matrix for $G$ if
\[
\left( \begin{array}{c|c} X_0 &X_{1} \\ \hline \noalign{\smallskip} 0 &G_{01} \end{array}\right)  =
\left( \begin{array}{c|c} -I &I \\ \hline \noalign{\smallskip} 0 &G \end{array}\right)_{01}.
\]
\item[(Y)] The matrix $Y=Y_1+Y_0$ is a right-ordered, left-reduced characteristic matrix for $H$ if
\[
\left( \begin{array}{c|c} H_{10} &0 \\ \hline \noalign{\smallskip} Y_1 &Y_0 \end{array}\right) = 
\left( \begin{array}{c|c} H &0 \\ \hline \noalign{\smallskip} I &-I \end{array}\right)_{10}
\]
\end{itemize}
\end{definition}

\begin{example}
\[
\begin{array}{ccm{2.25in}l}
\left[ \begin{array}{c|c} X_0 &X_1 \\ \hline 0 &G \end{array}\right] 
&=
&\begin{tikzpicture}
	[scale=.5,auto=center]
	\matrix [matrix of math nodes,ampersand replacement=\&,
                    left delimiter={[},right delimiter={]},row sep = .1em,column sep = .2em,text width=.6em,text height=.7em] (m)
        {
	2 \&\;2 \&\;1 \&\;0  \&\;\dt\&\;\dt\&\;\dt\&\;\dt  \\
        \dt\&\;1 \&\;1 \&\;1 \&\;0\&\;\dt\&\;\dt\&\;\dt \\
        \dt\&\;\dt\&\;1 \&\;2 \&\;1 \&\;0 \&\;\dt\&\;\dt \\
        \dt\&\;\dt\&\;\dt\&\;2 \&\;2 \&\;1\&\;0\&\;\dt \\ \hline
        \dt\&\;\dt\&\;\dt\&\;\dt \&\;2 \&\;2 \&\;1 \&\;0 \\
        \dt\&\;\dt\&\;\dt \&\;\dt \&\;\dt \&\;1 \&\;1 \&\;1\\
        };  
        \draw[color=black] (m-1-5.north west) -- (m-6-5.south west);
        \draw[color=black] (m-4-1.south west) -- (m-4-8.south east);
        \draw[color=black,dashed] (m-3-3.north west) -- (m-3-5.north west) -- (m-4-5.south west) -- (m-4-3.south west) -- (m-3-3.north west);
        \draw[color=black,dashed] (m-3-5.north west) -- (m-3-7.north west) -- (m-4-7.south west) -- (m-4-5.south west) -- (m-3-5.north west);
\end{tikzpicture}
& \\  \noalign{\smallskip}
&\perp
&\begin{tikzpicture}
	[scale=.5,auto=center]
	\matrix [matrix of math nodes,ampersand replacement=\&,
                    left delimiter={[},right delimiter={]},row sep = .1em,column sep = .2em,text width=.6em,text height=.7em] (m)
        {
	1\&\;1\&\;2\&\;0\&\;1\&\;1\&\;2\&\;\dt  \\
        \dt\&\;1\&\;1\&\;1\&\;0\&\;1\&\;1\&\;1 \\
        };  
        \draw[color=black] (m-1-5.north west) -- (m-2-5.south west);
        \draw[color=black,dashed] (m-1-3.north west) -- (m-1-5.north west) -- (m-2-5.south west) -- (m-2-3.south west) -- (m-1-3.north west);
        \draw[color=black,dashed] (m-1-5.north west) -- (m-1-7.north west) -- (m-2-7.south west) -- (m-2-5.south west) -- (m-1-5.north west);
\end{tikzpicture}
&=~ \left[ ~H~|~H~ \right].
\end{array}
\]
\[
\begin{array}{ccm{2.25in}l}
\left[ ~G~|~G~ \right]
&=
&\begin{tikzpicture}
	[scale=.5,auto=center]
	\matrix [matrix of math nodes,ampersand replacement=\&,
                    left delimiter={[},right delimiter={]},row sep = .1em,column sep = .2em,text width=.6em,text height=.7em] (m)
        {
          2 \&\;2 \&\;1 \&\;0  \&\;2\&\;2\&\;1\&\;\dt  \\
          \dt\&\;1 \&\;1 \&\;1 \&\;0\&\;1\&\;1\&\;1  \\	
        };  
        \draw[color=black] (m-1-5.north west) -- (m-2-5.south west);
        \draw[color=black,dashed] (m-1-3.north west) -- (m-1-5.north west) -- (m-2-5.south west) -- (m-2-3.south west) -- (m-1-3.north west);
        \draw[color=black,dashed] (m-1-5.north west) -- (m-1-7.north west) -- (m-2-7.south west) -- (m-2-5.south west) -- (m-1-5.north west);
\end{tikzpicture} 
& \\  \noalign{\smallskip}
&\perp
&\begin{tikzpicture}
	[scale=.5,auto=center]
	\matrix [matrix of math nodes,ampersand replacement=\&,
                    left delimiter={[},right delimiter={]},row sep = .1em,column sep = .2em,text width=.6em,text height=.7em] (m)
        {
	1 \&\;1 \&\;2 \&\;\dt  \&\;\dt\&\;\dt\&\;\dt\&\;\dt  \\
        0 \&\;1 \&\;1 \&\;1 \&\;\dt\&\;\dt\&\;\dt\&\;\dt \\ \hline
        \dt\&\;0\&\;1 \&\;2 \&\;1 \&\;\dt\&\;\dt\&\;\dt \\
        \dt\&\;\dt\&\;0\&\;1 \&\;1 \&\;2\&\;\dt\&\;\dt \\
        \dt\&\;\dt\&\;\dt\&\;0 \&\;1 \&\;1 \&\;2 \&\;\dt \\
        \dt\&\;\dt\&\;\dt \&\;\dt \&\;0 \&\;1 \&\;1 \&\;1\\
        };  
        \draw[color=black] (m-1-5.north west) -- (m-6-5.south west);
        \draw[color=black] (m-2-1.south west) -- (m-2-8.south east);
        \draw[color=black,dashed] (m-3-3.north west) -- (m-3-5.north west) -- (m-4-5.south west) -- (m-4-3.south west) -- (m-3-3.north west);
        \draw[color=black,dashed] (m-3-5.north west) -- (m-3-7.north west) -- (m-4-7.south west) -- (m-4-5.south west) -- (m-3-5.north west);
\end{tikzpicture}
&=~ \left[ \begin{array}{c|c} H &0 \\ \hline \noalign{\smallskip} Y_1 &Y_0 \end{array}\right]
\end{array}
\]
\end{example}

Let $G$ have left pivot columns $I_0$ and right pivot columns $I_1$. So that $H$ has left pivot columns $J_1 = \{ 0, 1, \ldots, n-1 \} \backslash I_1$ and right pivot columns $J_0 =  \{ 0, 1, \ldots, n-1 \} \backslash I_0$. 

\begin{proposition} \label{P:gen}
Let $G \perp H$ as in Lemma \ref{L:gen}.
\begin{itemize} 
\item[(X)] A matrix $X=X_0+X_1$ is a left-ordered characteristic matrix for $G$ if and only if the rows of $(X_0|X_1)$ are orthogonal to $(H|H)$ and have spans starting in distinct elements of $\{0,\ldots,n-1\}$ and ending in distinct elements of $I_1 \cup n+J_1$. 
\item[(Y)] A matrix $Y=Y_1+Y_0$ is a right-ordered characteristic matrix if and only if the rows of $(Y_1|Y_0)$ are orthogonal to $(G|G)$ and have spans starting in distinct elements of $I_1 \cup n+J_1$ and ending in distinct elements of $n+ \{0,\ldots,n-1\}$. 
\end{itemize}
\end{proposition}
\begin{proof} The matrix $(H|H)$ has spans starting in $J_1$. The orthogonal row space in minimal span form therefore has spans ending in $I_1 \cup  n+ \{0,\ldots,n-1\}$. With rows in $(0|G)$ ending in $n + I_1$, the orthogonal row space is in minimal span form if and only if the rows in $(X_0|X_1)$ end in $I_1 \cup n+ J_1.$ The proof for (Y) is similar. 
\end{proof}

\begin{corollary} \label{C:span} It follows that
\begin{itemize}
\item[(X)] The total spanlength of $(X_0|X_1)$ and thus of $X$ is $|J_1| n = (\rk H) n$.
\item[(Y)] The total spanlength of $(Y_1|Y_0)$ and thus of $Y$ is $|I_1| n = (\rk G) n$.
\end{itemize}
\end{corollary}

Definition \ref{D:red} refers to the minimal span forms of $(H|H)^\perp$ and $(G|G)^\perp$. A characteristic matrix can also be obtained directly from minimal span forms for the matrices $G$ and $H$. This characteristic matrix will in general not be reduced.  
The rows of a matrix $G$ in minimal span form provide rows for $(X_0|X_1)$ with spans starting in $I_0$ and ending in $I_1$. The remaining generators, with span starting in $J_0$ and ending in $n + J_1$, can be obtained from a matrix $H$ in minimal span form using the following lemma. The dual construction yields rows for $(Y_1|Y_0)$ with span starting in $I_1$ and ending in $n+I_0$. 

\begin{lemma} \label{L:sred}
Let $G \perp H$ be a pair of matrices with maximal orthogonal row spaces and assume that both matrices are in minimal span form.
\begin{itemize}
\item[(X)] Let $H'_0$ be the unique matrix of the same size as $H$ and with support in the trailing pivot columns $J_0$ of $H$ such that $H'_0 H^T = -I$ and let $H'_1$ be the unique matrix of the same size as $H$ and with support in the leading pivot columns $J_1$ of $H$ such that $H'_1 H^T = I$. The rows of $(H'_0|H'_1)$ are orthogonal to $(H|H)$ and have spans starting in $J_0$ and ending in $n+ J_1$.
\item[(Y)] Let $G'_0$ be the unique matrix of the same size as $G$ and with support in the leading pivot columns $I_0$ of $G$ such that $G'_0 G^T = -I$ and let $G'_1$ be the unique matrix of the sime size as $G$ and with support in the trailing pivot columns $I_1$ of $G$ such that $G'_1 G^T = -I$. The rows of $(G'_1|G'_0)$ are orthogonal to $(G|G)$ and have spans starting in $I_1$ and ending in $n + I_0$.
\end{itemize}
\end{lemma}
\begin{proof} 
Clear.
\end{proof}

\begin{example}
\[
    \begin{tikzpicture}
        [scale=.5,auto=center,every node/.style={minimum size=0em}]
        \node at (1,11.6) {$H ~= $};      
        \matrix [matrix of math nodes,left delimiter=(,right delimiter=)] (m) at (6,11.6)
        {
             1 &1 &2 &0  \\ 0 &1 &1 &1 \\ 
        };  
        \draw[dashed] (m-1-1.north west) -- (m-1-2.north east) -- (m-2-2.south east) -- (m-2-1.south west) -- (m-1-1.north west);
        \draw[dashed] (m-1-1.north west) -- (m-1-4.north east) -- (m-2-4.south east) -- (m-2-2.south east);
        \node at (10.5,11.6) {$H'_0 = $};
        \matrix [matrix of math nodes,left delimiter=(,right delimiter=)] (m) at (15,11.6)
        {
             0 &0 &1 &2  \\ 0 &0 &0 &2 \\ 
        };  
        \draw[dashed] (m-1-3.north west) -- (m-1-4.north east) -- (m-2-4.south east) -- (m-2-3.south west) -- (m-1-3.north west);
        \node at (20,11.6) {$H'_1 = $};
        \matrix [matrix of math nodes,left delimiter=(,right delimiter=)] (m) at (24.5,11.6)
        {
             1 &0 &0 &0  \\ 2 &1 &0 &0 \\ 
        };  
        \draw[dashed] (m-1-1.north west) -- (m-1-2.north east) -- (m-2-2.south east) -- (m-2-1.south west) -- (m-1-1.north west);
    \end{tikzpicture}
\]
\[
    \begin{tikzpicture}
        [scale=.5,auto=center,every node/.style={minimum size=0em}]
        \node at (1,11.6) {$G ~= $};      
        \matrix [matrix of math nodes,left delimiter=(,right delimiter=)] (m) at (6,11.6)
        {
             2 &2 &1 &0  \\ 0 &1 &1 &1 \\ 
        };  
        \draw[dashed] (m-1-1.north west) -- (m-1-2.north east) -- (m-2-2.south east) -- (m-2-1.south west) -- (m-1-1.north west);
        \draw[dashed] (m-1-1.north west) -- (m-1-4.north east) -- (m-2-4.south east) -- (m-2-2.south east);
        \node at (10.5,11.6) {$G'_1 = $};
        \matrix [matrix of math nodes,left delimiter=(,right delimiter=)] (m) at (15,11.6)
        {
             0 &0 &1 &2  \\ 0 &0 &0 &1 \\ 
        };  
        \draw[dashed] (m-1-3.north west) -- (m-1-4.north east) -- (m-2-4.south east) -- (m-2-3.south west) -- (m-1-3.north west);
        \node at (20,11.6) {$G'_0 = $};
        \matrix [matrix of math nodes,left delimiter=(,right delimiter=)] (m) at (24.5,11.6)
        {
             1 &0 &0 &0  \\ 1 &2 &0 &0 \\ 
        };  
        \draw[dashed] (m-1-1.north west) -- (m-1-2.north east) -- (m-2-2.south east) -- (m-2-1.south west) -- (m-1-1.north west);
    \end{tikzpicture}
\]
\end{example}
 
Complementing the rows of $G$ with those of $H'$ we obtain a full set of characteristic generators for $G$.

\begin{theorem} \label{T:sred}
For matrices $G \perp H$ in minimal span form, let $G'_0, G'_1, H'_0, H'_1$ be as in Lemma \ref{L:sred}. So that
\[
G'_0 G^T = -I, ~~~~G'_1 G^T = I, ~~~~H'_1 H^T = I, ~~~~H'_0 H^T = -I.
\]
Let
\[
\begin{array}{rrr}
( X_0 | X_1 ) = \left[ \begin{array}{c|c} G &0 \\ H'_0 &H'_1 \end{array} \right], 
&~~~~X = X_0 + X_1 =  \left[ \begin{array}{c} G \\ H' \end{array} \right]. \\ \noalign{\medskip}
( Y_1 | Y_0 ) = \left[ \begin{array}{c|c} G'_1 &G'_0 \\ 0 &H \end{array} \right], 
&~~~~Y = Y_1 + Y_0 =  \left[ \begin{array}{c} G' \\ H \end{array} \right].
\end{array}
\]
Then $X$ and $Y$ are characteristic matrices for $G$ and $H$. Moreover,
\begin{align*}
X_0 Y_0^T = - I, ~~X_0 Y_1^T + X_1 Y_0^T = I, ~~X_1 Y_1 ^T = 0. \\
Y_0^T X_0 = - I, ~~Y_0^T X_1 + Y_1^T X_0 = I, ~~Y_1^T X_1 = 0.
\end{align*}
So that $X Y^T = 0$ and $Y^T X = 0.$
\end{theorem}
\begin{proof}
Lemma \ref{L:sred} in combination with Proposition \ref{P:gen} shows that $X$ and $Y$ are characteristic matrices. The two groups of three equalities hold by straightforward verification and adding the equalities in each group yields $X Y^T = 0$ and $Y^T X = 0.$
\end{proof}
 
In \cite{KV03}, the rows of $G$ are called conventional generators and the remaining generators are called circular generators. The latter generators correspond to conventional generators in a cyclicly shifted version of the matrix $G$. In Lemma \ref{L:sred} we rely on the matrix $H$ to construct the circular generators and make no reference to the cyclicly shifted versions of the matrix $G$. The circular generators correspond to the rows of a matrix $H'$.

\begin{proposition} \label{P:abcd}
Let the $n$ rows of $(X_0|X_1)$ be orthogonal to $(H|H)$ with $X_0$ upper triangular and nonzero on the diagonal and $X_1$ strictly lower triangular. The following are equivalent.
\begin{itemize}
\item[(a)] The rows of $(X_0|X_1)$ end in the distinct columns $I_1 \cup n + J_1$.
\item[(b)] The total spanlength of $(X_0|X_1)$ is $|J_1| n.$
\item[(c)] The number of rows with span in $(0,a-1]$ or in $(a,a+n-1]$ is $|I_1|$, for each $a \in \{0, 1, \ldots, n-1 \}$.   
\item[(d)] The number of rows with $(a-1,a]$ or $(a+n-1,a+n]$ in their span is $|J_1|$, for each $a \in \{0, 1, \ldots, n-1 \}$.   
\end{itemize}
\end{proposition}
\begin{proof}
The equivalence between (c) and (d) is clear, as are the implications (a) implies (b) and (d) implies (b) . For $(X_0|X_1)$ such that (b) holds, the rows of $(X_0|X_1)$ are part of a matrix in minimal span form for $(H|H)^\perp$ and, as is in Proposition \ref{P:gen}, this implies (a). The rows under (c) give an independent set of relations among columns of $H$. Thus their number is at most $|I_1|$, for each $a \in \{0, 1, \ldots, n-1 \}$. The number of rows counted under (d) is therefore at least $|J_1|$, for each $a \in \{0, 1, \ldots, n-1 \}$. Equality holds under (c) and (d) for all $a \in \{0, 1, \ldots, n-1 \}$ if and only if (b).
\end{proof}

The proposition may be used to verify that slightly different definitions for the characteristic matrix in \cite{KV03}, \cite{GW11a} and Definition \ref{D:char} are indeed equivalent. By Proposition \ref{P:gen}, a characteristic matrix $X$ as in Definition \ref{D:char} is characterized by properties (a) and (b). The definitions \cite{KV03} and \cite{GW11a} make use of properties (c) and (d), respectively. Property (c) is used in \cite{KV03} to construct a characteristic matrix from conventional generators of cyclicly permuted versions of $G$. Property (d) is used in \cite{GW11a} as an axiomatic property, from which further properties can be obtained without reference to a specific construction of the characteristic matrix.

\begin{lemma} \label{L:a}
For a characteristic matrix $X$ and for each $a \in \{ 0,1,\ldots,n-1\}$, the $|I_1|$ rows of $X$ with span in $(a,a-1]$ are linearly independent and they form a basis for the row space of $X$.
\end{lemma}

\begin{proof}
As in the proof of the proposition, the number of such rows is $|I_1|$ and after unwrapping it is clear that they specify independent relations among the columns of $H$.
\end{proof}   

The selection of rows in the lemma has an interpretation as a minimal span basis for the matrix $G$ cyclicly shifted by $a$ columns.

\begin{example} For the pair of binary orthogonal matrices $G$ and $H$ in Section \ref{S:reds}, Appendix \ref{S:binex} gives the reduced characteristic matrices $X$ and $Y$. The matrices $G_{01}$ and $H_{10}$ appear as submatrices in $X$ and $Y$. Together they explain all the entries except those in the dashed boxes. The entries in the dashed boxes are in duality witht the corresponding full boxes via Lemma \ref{L:sred}.
\end{example} 

\subsection{Subspaces defined by intervals}

We mention another combinatorial characterization of matrices in minimal span form and characteristic matrices. In each case the rows of the matrix can be used to distinguish between subspaces defined with different support intervals. The situation is as in Boolean Information Retrieval (BIR), where documents can be distinguished based on different matches with a set of index terms. A set of index terms separates the documents in a collection if each document has a unique subset of matching index terms. Similarly, we say that a subset $B$ of vectors separates a collection of subspaces if for any two subspaces $V$ and $V'$ in the family
\begin{equation} \label{E:bir}
V = V' ~~\Leftrightarrow~~ V \cap B = V' \cap B.
\end{equation}
We apply this to a matrix $G$ and its row space $C = \row G$. For a subset $I  \subset \{ 0,1,\ldots,n-1 \}$ of column indices, let $C(I) = \{ c \in C : \text{$c_k = 0$ for $k \not \in I$} \}$. We say that the subspace $C(I)$ is supported on an interval if $I = \{ i : a \leq i \leq b \}$ for some $0 \leq a < b \leq n-1.$

\begin{lemma} \label{L:sep1}
The rows in a matrix $G$ in minimal span form separate the collection of subspaces $\{ C(I) : \text{$I$ an interval} \}$. 
\end{lemma}

\begin{lemma} \label{L:sep2}
The rows in a characteristic matrix $X$ for $G$ separate the collection of subspaces $\{ C(I) :  \text{$I$ an interval} \} \cup \{ C(J) :  \text{$J$ a complement of an interval} \}.$ 
\end{lemma}

In the case of $G$, the rows moreover form a basis for the row space and it follows that the rows have the \emph{subsystem basis property} \cite[Theorem~3]{F11}. The separating properties of characteristic generators are the same as those of \emph{discrepancies} that are used in \cite{DPdelta}, \cite{Dgk} to distinguish algebraic functions with a specified behavior at the origin and at infinity. After writing a vector $(c_0,c_1,\ldots,c_{n-1})$ as a polynomial $c_0 + c_1 x + \cdots c_{n-1} x^{n-1}$, the vector has span $(i,j]$ if and only if the polynomial has a pole of order $j$ at infinity and a zero of order $i$ at the origin. Whereas the characteristic spans of a row space provide information about its trellis structure, the discrepancies of a function field give lower bounds for the minimum distance of codes constructed with the function field \cite{DKPjpaa}. We refer to the listed references for details.

\begin{definition}
Define a bijection $\sigma$ on $\{0, 1, \ldots, n-1 \}$ such that, for $i \in I_0$, $\sigma(i) \in I_1$ and $(i,\sigma(i)]$ is a minimal span for $G$, and, for $j \in J_0$, $\sigma(j) \in J_1$ and $(\sigma(j),j]$ is a minimal span for $H$. 
Thus, $\{ (i,\sigma(i)] : i = 0,1,\ldots,n-1 \}$ is the collection of spans for $X$ and $\{ (\sigma(j),j] : j = 0,1,\ldots,n-1 \}$ is the collection of spans for $Y$.
\end{definition}

In the interpretation of \cite[Section V, Figure 14]{KS95}, the set $\{ (i,\sigma(i)] : i = 0,1,\ldots,n-1 \}$ defines the positions of $n$ nonattacking rooks on a chess board of size $n$, divided into $|I_1|$ black rooks above the main diagonal and $|J_1|$ white rooks below the main diagonal. Appendix \ref{S:rooks} shows a rook configuration from a Suzuki function field.
 
\bigskip
 
A slightly different version of a subspace separating subset is a stopping set used in erasure decoding. Let $I \cup J$ be a partition of the coordinates and assume that the vector $c$ is known in the positions $I$ and erased in the positions $J$. A subset $S$ of the dual code $D$ is a stopping set for $J$ if it iscapable of recovering at least one of the erased symbols in $J$.
\begin{equation} \label{E:stop}
(\exists j \in J)~D(I \cup j) \neq D(I) ~~\Leftrightarrow~~  (\exists j \in J)~D(I \cup j) \cap S \neq D(I) \cap S.
\end{equation} 


\subsection{Unit memory convolutional codes}

The decomposition $X = X_0+ X_1$ has an interpretation as the \emph{unwrapping} of the matrix $X$ into two matrices $X_0$ and $X_1$ that define a unit-memory convolutional code via $y_t = u_t X_0 + u_{t-1} X_1$, with $y_t$ and $u_t$ vectors of length $n$ indexed by $t \geq 0.$ 
\[
\left[ \begin{array}{ccccccc}
X_0 &X_1 & & & \\
 &X_0 &X_1 & & \\
 & &\ddots  &\ddots & 
\end{array} \right] \qquad
\left[ \begin{array}{ccccccc}
Y_0 & & & & \\
 Y_1 &Y_0 & & & \\
 &\ddots  &\ddots & 
\end{array} \right]
\]

For the codes below ($X$ on the left and $Y$ on the right) we extend the time  $t$ to all integers. 
\[
{\small
\begin{array}{cccc|cccccccc}
2 &2 &1 &0  &\dt&\dt&\dt&\dt  &\dt&\dt&\dt&\dt \\
\dt&1 &1 &1 &0&\dt&\dt&\dt&\dt&\dt&\dt&\dt \\
\dt&\dt&1 &2 &1 &0 &\dt&\dt&\dt&\dt&\dt&\dt \\
\dt&\dt&\dt&2 &2 &1&0&\dt&\dt&\dt&\dt&\dt \\ \hline

\dt&\dt&\dt&\dt&2 &2 &1 &0&\dt&\dt&\dt&\dt \\
\dt&\dt&\dt&\dt&\dt&1 &1 &1 &0&\dt&\dt&\dt \\
\dt&\dt&\dt&\dt&\dt&\dt&1 &2 &1&0 &\dt&\dt \\
\dt&\dt&\dt&\dt&\dt&\dt&\dt&2 &2 &1&0&\dt \\ 

\end{array}
\quad \quad
\begin{array}{cccccccc|cccc}

\dt&0&1&2 &1 &\dt &\dt &\dt&\dt&\dt&\dt&\dt \\
\dt&\dt&0&1 &1&2 &\dt &\dt &\dt&\dt&\dt&\dt \\
\dt&\dt&\dt&0 &1&1&2 &\dt &\dt&\dt &\dt&\dt \\
\dt&\dt&\dt&\dt &0&1&1&1 &\dt &\dt&\dt&\dt \\ \hline

\dt&\dt&\dt&\dt &\dt&0&1&2  &1 &\dt &\dt &\dt \\
\dt&\dt&\dt&\dt &\dt&\dt&0&1 &1&2 &\dt &\dt \\
\dt&\dt&\dt&\dt &\dt&\dt&\dt&0 &1&1&2 &\dt \\
\dt&\dt&\dt&\dt &\dt&\dt&\dt&\dt  &0&1&1&1 \\

\end{array}
}
\]

\section{Duality for characteristic matrices} \label{S:dual}

We first collect properties that hold for any pair of characteristic matrices $X$ and $Y$. When we write $X = X_0 + X_1$ or $Y = Y_1 + Y_0$ it is implicitly understood that $X_0$ is upper triangular and $X_1$ strictly lower triangular, and $Y_1$ strictly upper triangular and $Y_0$ lower triangular. 

\begin{lemma} \label{L:minor}
For any pair of characteristic matrices $X$ and $Y$, and for $X=X_0+X_1$ and $Y=Y_1+Y_0$, the row spaces of $X$ and $Y$ are orthogonal and the columns spaces of $X_1$ and $Y_1$ are orthogonal.
\begin{equation} \label{E:0}
X Y^T = 0 ~~~\text{and}~~~ X_1^T Y_1 = 0.
\end{equation}
The matrix $X_1$ has a full rank minor in the positions $J_0 \times J_1$. The matrix $Y_1$ has a full rank minor in the positions $I_0 \times I_1$. For reduced characteristic matrices and for semi-reduced characteristic matrices, the entries outside these minors are zero. 
\end{lemma}
\begin{proof}
The row spaces of $X$ and $Y$ are those of the matrices $G$ and $H$ and $X Y^T = 0$ follows from $G H^T =0$. The matrix $X_1$ has zero rows in the positions $I_0$ and the matrix $Y_1$ has zero rows in the complementary positions $J_0$. The nonzero rows of $X_1$ in the positons $J_0$ end in distint positions of $J_1$.
The nonzero rows of $Y_1$ in the positions $I_0$ start in distinct positions of $I_1$. The last claim follows from Definition \ref{D:red} and Theorem \ref{T:sred}.
\end{proof}

The characteristic matrices $X$ and $Y$ in Theorem \ref{T:sred} satisfy the longer list of properties
\begin{align}
X_0 Y_0^T = - I, ~~X_0 Y_1^T + X_1 Y_0^T = I, ~~X_1 Y_1 ^T = 0., \label{E:1x} \\
Y_0^T X_0 = - I, ~~Y_0^T X_1 + Y_1^T X_0 = I, ~~Y_1^T X_1 = 0, \label{E:1y}
\end{align}
and
\begin{equation} \label{E:d}
Y^T X = 0.
\end{equation}
We show that (\ref{E:1x}), (\ref{E:1y}) and (\ref{E:d}) all follow from (\ref{E:0}) and
\begin{equation} \label{E:cd}
X_0 Y_1^T + X_1 Y_0^T~=~I.
\end{equation}

\begin{lemma} \label{L:ass1}
Let $X = X_0 + X_1$, with $X_0$ upper triangular and $X_1$ strictly lower triangular, and let $Y = Y_1 + Y_0$, with $Y_1$ strictly upper triangular and $Y_0$ lower triangular. If (\ref{E:0}) and (\ref{E:cd}) hold then also (\ref{E:1x}), (\ref{E:1y}) and (\ref{E:d}) hold. From (\ref{E:1x}) it follows that
\begin{equation} \label{E:mx}
\left( \begin{array}{c|c} X_0 &X_1 \\ \hline \noalign{\smallskip} 0 &X_0 \end{array}\right)
\left( \begin{array}{c|c} Y_0 &0 \\ \hline \noalign{\smallskip} Y_1 &Y_0 \end{array}\right)^T = \left( \begin{array}{r|r} -I &I \\ \hline \noalign{\smallskip} 0 &-I \end{array}\right)
\end{equation}
\end{lemma}
\begin{proof}
From $X Y^T = 0$ and $X_0 Y_1^T + X_1 Y_0^T = I$, we obtain $X_0 Y_0^T + X_1 Y_1 ^T = -I.$ Using that $X_0 Y_0^T$ is upper triangular and $X_1 Y_1^T$ is strictly lower triangular (\ref{E:1x}) and (\ref{E:mx}) follow. After writing (\ref{E:mx}) as
\[
\left( \begin{array}{c|c} X_0 &X \\ \hline \noalign{\smallskip} 0 &X_0 \end{array}\right)
\left( \begin{array}{c|c} Y_0^T &Y_1^T \\ \hline \noalign{\smallskip} 0 &Y_0^T \end{array}\right)
= \left( \begin{array}{c|c} -I &0 \\ \hline \noalign{\smallskip} 0 &-I \end{array}\right).
\]
the matrices on the left commute and (\ref{E:1y}) follows. Finally, after expanding $Y^T X$, $Y^T X = 0$ follows from (\ref{E:1y}). 
\end{proof}

The equality $X Y^T = 0$ is clear. The row spaces of $X$ and $Y$ are those of the matrices $G$ and $H$ and are thus orthogonal. The equality $Y^T X = 0$ on the other hand, i.e. orthogonality of the column spaces of $X$ and $Y$, is a property that does not hold for general choices of a characteristic matrix $X$ for $G$ and a characteristic matrix $Y$ for $H$. It has the following trivial but important consequence for trellis constructions. 

\begin{lemma} \label{L:XY}
Assume that $X$ and $Y$ have maximal orthogonal column spaces. A selection of rows in $X$ forms a basis for the row space of $X$ if and only if the complementary rows in $Y$ form a basis for the row space of $Y$.
\end{lemma}
\begin{proof} We avoid relying on matroid duality and write out the short proof: The selected rows in $X$ form a basis if and only if there exists no nonzero column in $X$ with zeros in the selected rows if and only if the complementary rows in $Y$ form a basis. 
\end{proof}

The property $Y^T X = 0$, i.e. the orthogonality of column spaces, does not hold in general and is a conseqeunce of $X_0 Y_1^T + X_1 Y_0^T = I$. The latter expresses a duality between rows of $(X_0|X_1)$ and rows of $(Y_1|Y_0)$, inner products are $1$ for rows in the same position and $0$ otherwise. The rows of $(X_0|X_1)$ form a basis for the coset scheme $(H|H)^\perp / (G|G)$ and those of $(Y_1|Y_0)$ a dual basis for the dual coset scheme $(G|G)^\perp / (H|H).$ 
 
\begin{definition} We say that characteristic matrices $X=X_0+X_1$ and $Y=Y_1+Y_0$ are in duality if $X_0 Y_1^T + X_1 Y_0^T = I$. In other words if, in addition to $X$ and $Y$ being characteristic matrices for the row spaces of orthogonal matrices $G$ and $H$, the row spaces
of $(X_0|X_1)$ and of $(Y_1|Y_0)$ are in duality as bases for the coset schemes $(H|H)^\perp / (G|G)$ and $(G|G)^\perp / (H|H).$  
\end{definition}

The definition is  justified by Equation (\ref{E:mx}) in Lemma \ref{L:ass1}, which shows that any characteristic matrix has a unique dual characteristic matrix, such that a left-ordered characteristic matrix is in duality with a right-ordered characteristic matrix. We still need to establish that this duality is the one that corresponds to trellis duality. This is done in the following theorem.

\begin{theorem} \label{T:trdual}
Trellises constructed via the Nori-Shankar span based BCJR construction from dual characteristic matrices $X$ and $Y$ are in duality.
\end{theorem}
\begin{proof}
The first part of the claim is the rank duality: for a selection of rows in $X$ that form a basis for the row space of $X$, the rows in the complementary positions in $Y$ form a basis for the row space of $Y$. This was shown in Lemma \ref{L:XY} as an immediate consequence of $X^T Y = 0$. The second part is the duality of displacement matrices. We have to prove this duality for every pair of dual trellises obtained from $X$ and $Y$. Displacement matrices are submatrices of $X Y_1^T = X_0 Y_1^T =: D$ and of 
$Y X_1^T = Y_0 X_1^T =: E$. For $I$ a set of independent rows in $X$, the complement $J$ is a set of independent rows in $Y$ (by rank duality). The displacement matrices for the corresponding trellises are the submatrices $D|_{I \times J}$ and $E|_{J \times I}$. From $I \cap J = \emptyset$ and $D+E^T = X_0 Y_1^T + X_1 Y_0^T = I$ we see that
\begin{equation} \label{E:nsd}
E [J \times I]^T = -D[I \times J]
\end{equation}
which expresses duality of two tail-biting trellises as BCJR trellises.
\end{proof}

The presence of the minus sign in (\ref{E:nsd}) means that if we decide to use $E$ and $D$ to label vertices then syndromes will be accumulated in opposite directions for a BCJR trellis and its dual. A different approach to construct dual characteristic matrices and to prove trellis duality can be found in \cite{GW11b}. A core lemma in that paper is the following, stated here in a different form using the notation and concepts of this paper. We have also changed the roles of $X$ and $Y$.

\begin{lemma}(\cite[Lemma 4.5]{GW11b}) Given a characterisitc matrix $Y$ in right-ordered form, the top row $c_0$ of the dual characteristic matrix $X$ (and thus the top row of the matrix $X_0$) is obtained as a solution to the system of equations
\[
Y_0 \; c_0^T = (1,0,\ldots,0)^T.
\]
\end{lemma}

Together with \cite[Proposition 4.6]{GW11b}, which states that every other row of $X$ can be computed as the top row of a modified characteristic matrix determined by $X$, this yields a construction for the dual characteristic matrix. Clearly, in the given form the equation that is used to compute the top row of $X$ matches $Y_0 X_0^T = -I$ (as before up to a change of sign in the scaling). The proof in \cite{GW11b} that the constructed dual characteristic matrix yields dual trellises is  in several steps \cite[Theorem 4.8, Proposition 4.9, Proposition 4.11]{GW11b}. 
 
\begin{lemma} \label{L:DE}
The decomposition $I = X_0 Y_1^T + X_1 Y_0^T$ in (\ref{E:cd}) is a decomposition of orthogonal idempotents.
For $D = X_0 Y_1^T$ and $E = Y_0 X_1^T$, $I = D+E^T$,
\[ 
D^2 = D, ~E^2 = E, ~DE^T = 0 ~\text{and}~E^T D = 0. 
\]
It holds that $X = D X_0$ and $Y= E Y_0$. 
\end{lemma}

\begin{proof}
The orthogonality $DE^T = 0$ follows from $Y_1^T X_1 =0$ in (\ref{E:0}) and $E^TD = X_1(Y_0^T X_0)Y_1^T= 0$ from (\ref{E:1x}) and (\ref{E:1y}). 
For the relations involving $D$ use
\begin{align*}
&X_0 = X_0 Y_1^T X_0 + X_1 Y_0^T X_0 = X_0 Y_1^T X_0 - X_1. \\
&X = X_0 + X_1 =  X_0 Y_1^T X_0 = D X_0. \\
&D = X_0 Y_1^T = X Y_1^T = D X_0 Y_1^T = D^2.
\end{align*}
The relations involving $E$ follow in the same way.
\end{proof}


\subsection{On a conjecture by Koetter and Vardy}

In the previous section we defined a duality for characteristic matrices and showed that it is compatible with both local trellis duality (where it agrees with the duality formulated in \cite{GW11a}) and with BCJR trellis duality (where it agrees with the displacement matrix duality from \cite{NS06}).
The dualities provide a bijection between characteristic matrices for $G$ and those for its orthogonal matrix $H$. Koetter and Vardy had earlier proposed that there might be a duality between two special choices of characteristic matrices. In this section we adjust and then prove their proposed duality.

\bigskip

 In \cite{KV03} it was proposed to choose the lexicographically first one among generators with the same characteristic span. This defines a unique set of characteristic generators for a row space. It was then conjectured that the lexicographically first choices for a row space and its orthogonal dual would be in duality as characteristic matrices. Since then, examples have been found by Gleussen-Larssing and Weaver \cite{GW11a}, \cite{GW11b} for which the claimed duality fails. Moreover, the authors provide a detailed duality that matches each characteristic matrix to a unique dual characteristic matrix. Both the construction of a dual characteristic matrix and the proof that it is dual to the original matrix are given in several steps and involve many technical details. In the previous section we used the unwrapped versions $(X_0|X_1)$ and $(Y_1|Y_0)$ to formulate an explicit duality combined with a shorter proof.  

\bigskip

The original conjecture and subsequent work ignores the different direction of the time-axis in a trellis and its dual. After taking into account that lexicographically ordered has opposite meanings for a trellis and its dual trellis the conjecture is correct. The notion of lexicographical order that we use is any ordering of a field where the $0$ element is minimal combined with an ordering of polynomials $\sum_k c_k x^k < \sum_k c'_k x^k$ whenever $c_k < c'_k$ for the largest $k$ with $c_k \neq c'_k$. Row vectors in the expanded matrix $(X_0|X_1)$ follow the ordering $(c_0, c_1, \ldots)$ but row vectors in the dual expanded matrix $(Y_1|Y_0)$ follow the reverse ordering $(\ldots, d_1, d_0)$. In both cases we order rows as polynomials. This leads to a selection of minimal characteristic generators that corresponds to that of a matrix in reduced minimal span form. We replace the original conjecture with the following theorem.

\begin{theorem} \label{T:conj}
Reduced characteristic matrices are in duality.
\end{theorem}
\begin{proof}
For a characteristic matrix $X$, the unwrapped form $(X_0|X_1)$ is in minimal span form with trailing pivots in the positions $I_1 \cup n+J_1$ and with full minor $A = (X_0 | X_1) | _{I_1 \cup n+J_1}$ of the form $A  = PL .$ A reduced characteristic matrix $X$ has $X_1$ with support in the positions $J_0 \times J_1$, where it has a minor of full rank (Lemma \ref{L:minor}). The dual characteristic matrix $Y$, with unwrapped form $(Y_1|Y_0)$, satisfies $X_1 Y_1^T = Y_1^T X_1 =0$ and therefore $Y_1$ has support in the positions $I_0 \times I_1$. The matrix $(Y_1|Y_0)$ is in minimal span form with leading pivots in the positions $I_1 \cup n+J_1$ and with full minor $B = (Y_1 | Y_0) | _{I_1 \cup n+J_1}$ of the form $B = P'U .$ The matrix $(X_0|X_1)$ is zero in the columns $n+I_1$ and the matrix $(Y_1|Y_0)$ is zero in the columns $J_1$. Thus $(X_0|X_1)(Y_1|Y_0)^T=I$ becomes $AB^T = I$. So that $P'=P$ and $U=L^{-T}$. Moreover, since $X$ is reduced, $A = L'P$. But then $AB^T=I$ shows that $B=L'^{-T} P = U' P$ and thus $Y$ is reduced.
\end{proof}

\subsection{Column space characterization of duality}

We give a characterization of duality that refers only to $X$ and $Y$ and not to their decompositions as sums of triangular matrices.

\begin{lemma}
Different characteristic matrices $X$ and $X'$ for the same row space $G$ have different column spaces.
\end{lemma}
\begin{proof}
Let $Y$ be the dual characteristic matrix for $X$, so that columns of $X$ and $Y$ are orthogonal. With Lemma \ref{L:gen}, the matrices
\[
\left( \begin{array}{c|c} X_0 &X_{1} \\ \hline \noalign{\smallskip} 0 &G \end{array}\right)  \sim
\left( \begin{array}{c|c} X'_0 &X'_{1} \\ \hline \noalign{\smallskip} 0 &G \end{array}\right)  
\]
are row equivalent, are both in left-ordered minimal span form, and share the same set of spans. One, $X'$, can be obtained from the other, $X$, via a set of upper triangular row operations. Only row operations are allowed that preserve the minimal span form. Thus a row operation that adds a word with span $(i,j]$ to another word is allowed only if the other word has span $(i',j']$ such that $(i,j] \subset (i',j']$. A change in column $a$ affects only rows that contain $a$ in their span. The corresponding rows in $Y$ have spans that do not contain $a$. Thus these rows in $Y$ are linearly independent by Lemma \ref{L:a}. And any difference between a column of $X$ and a column of $X'$ does not give a relation among rows of $Y$. This implies that any column of $X'$ that is different from the corresponding column of $X$ is not orthogonal to the column space of $Y$ and therefore does not belong to the column space of $X$. 
\end{proof}     

With the lemma, given a characteristic matrix $Y$ for $H$, there exists a unique characteristic matrix $X$ for $G$ such that $Y^T X = 0$.

\begin{theorem} \label{T:col}
A pair of characteristic matrices $X$ and $Y$, with maximal orthogonal row spaces, is in duality if and only if $Y^T X = 0$, that is if and only if $X$ and $Y$ have orthogonal column spaces.
\end{theorem}
\begin{proof}
Duality is defined by $X_0 Y_1^T +X_1 Y_0^T = I$ and this implies that $Y^T X =0$. This proves the only if part. The if part follows from the lemma.
\end{proof}

\subsection{Displacement matrices} \label{S:disp} 

We give further properties of the matrices $D = X_0 Y_1^T$ and $E = Y_0 X_1^T$, that have a role as displacement matrix in the BCJR constrction of \cite{NS06}: Let the rows $I$ of $X$ be chosen as edge labelings for generators of a label code, and let the complementary rows $J$ of $Y$ be chosen as edge labelings for generators of a dual label code. For row $i$ of $X$, $i \in I$, the corresponding row $i$ of the submatrix $E[J \times I]^T$ is selected as the initial vertex label of the path. Similarly, for row $j$ of $Y$, $j \in J$, the corresponding row $j$ of the submatrix $D[I \times J]^T$ is selected as the initial vertex label of the path. The labeling of the remaining vertices in the label code proceeds as in the BCJR construction via the accumulation of partial syndromes. 

\bigskip

For characteristic matrices $X$ and $Y$, $X Y^T = 0$ and $X_1^T Y_1 = 0$. Moreover, for matrices in duality, $X^T Y = 0 $ and $X_1 Y_1^T = 0$. 
In general the minors $J_0 \times J_1$ of $X_1$ and $I_0 \times I_1$ of $Y_1$ are invertible. The matrix $X_1$ is zero in rows $I_0$, and the matrix $Y_1$ is zero in rows $J_0$. We write this as
\[
D_0 X_1 = 0, ~~~E_0 Y_1 = 0,
\] 
for diagonal $0,1$-matrices $D_0$ and $E_0$, with $D_0 = 1$ in the positions $I_0$ and $E_0 =1$ in the positions $J_0$. So that $D_0 + E_0 = I$.

\begin{lemma} \label{L:disp1}
The nonzero columns $I_0$ of $D$ form a basis in systematic form for the column space of $X$ and the nonzero columns $J_0$ of $E$ form a basis in systematic form for the column space of $Y$. Equivalently,  
\begin{align}
D D_0 = D,  ~~~&D_0 D = D_0, ~~~D^T Y = 0  \label{Eq:D} \\
E E_0 = E, ~~~&E_0 E = E_0, ~~~E^T X = 0  \label{Eq:E}
\end{align}
\end{lemma}
\begin{proof} 
\begin{align*}
&D D_0 = D (I - E_0)  = D - X_0 Y_1^T E_0 = D, \\
&D_0 D = D_0 (I - X_1 Y_0^T) = D_0, \\
&(D^T) Y = Y_1 X^T Y = 0.
\end{align*}
Similar for (\ref{Eq:E}).  
\end{proof}

From (\ref{Eq:D}) and (\ref{Eq:E}) we see that 
\begin{equation} \label{Eq:dd0}
X = (D+E^T)X = D X = D (D_0 X).
\end{equation}
Here $D_0 X$ has as nonzero rows the conventional generators of $X$.  And another interpretation for $D$ is that it describes the rows of $X$ as linear combinations of its conventional generators.

\bigskip

Let $S$ be the permutation matrix that shifts row vectors to the right and column vectors upwards: $e_0 S  = e_1, \ldots, e_{n-1} S = e_0,$ and 
$S e_1^T = e_0^T,  \ldots, S e_0^T = e_{n-1}^T.$ Let $X^S = S X S^T$ be the conjugate of a left-ordered characteristic matrix $X$ for $G$. Then $X^S$ is a left-ordered characteristic matrix for $G S^T$. If $Y$ is a right-ordered characteristc matrix for $H$ then $Y^S$ is a right-ordered characteristic matrix for $H S^T$. For $D(X) = X_0 Y_1^T  = X X_0^{-1}$ it is in general not true that $D(X^S) = D(X)^S$.

\begin{lemma} \label{L:disp2}
Let $X$ have left column $x$ and $Y$ left column $y$. Then 
\[
D(X^S)^{S^T} = D(X) + x y^T, \qquad E(Y^S)^{S^T} = E(Y) - y x^T.
\]
\end{lemma}
\begin{proof}
Let $Y = \left( \begin{array}{c|c} d &d_0 \\ \hline \noalign{\smallskip} d_1 &B \end{array}\right).$ Then 
\begin{align*}
D(X) = X Y_1^T = X \left( \begin{array}{c|c} 0 &d_0 \\ \hline \noalign{\smallskip} 0 &B_1 \end{array}\right)^T, ~~~~
&Y^S =  \left( \begin{array}{c|c} B &d_1 \\ \hline \noalign{\smallskip} d_0 &d \end{array}\right),
\end{align*}
And
\begin{align*}
D(X^S)^{S^T} &= (X^S)^{S^T} ((Y^S)_1^T)^{S^T} = X  S^T \left( \begin{array}{c|c} B_1 &d_1 \\ \hline \noalign{\smallskip} 0 &0 \end{array}\right)^T S =  
X \left( \begin{array}{c|c} 0 &0 \\ \hline \noalign{\smallskip} d_1 &B_1 \end{array}\right)^T \\
&= D(X) + X \left( \begin{array}{c|c} d &0 \\ \hline \noalign{\smallskip} d_1 &0 \end{array}\right)^T - X \left( \begin{array}{c|c} d &d_0 \\ \hline \noalign{\smallskip} 0 &0 \end{array}\right)^T = D(X) + x y^T.
\end{align*}
The claim for $E(Y^S)$ follows with $D+E^T = I.$
\end{proof}

\begin{theorem} \label{T:disp}
Let $S(Y^T|X)$ be the label code representing the rows of $X$ as a BCJR trellis with vertex labelings from the column space of $Y$. 
Let $N_0, N_1, \ldots, N_{n}$ be the submatrices of $S(Y^T|X)$ representing the vertex labelings. The characterization of $N_0 = E^T$ 
by Lemma \ref{L:disp1} as a matrix with nonzero rows in systematic form orthogonal to columns of $X$ applies to all $N_0, N_1, \ldots, N_{n}$.
\end{theorem}

\begin{example}
\begin{align*}
Y^T = &\left[ \begin{array}{cccc}
1 &1 &1 &0 \\
0 &2 &1 &1 \\
1 &0 &2 &1 \\
2 &1 &0 &1 
\end{array} \right] ~~~~~~
X = \left[ \begin{array}{cccc}
2 &2 &1 &0 \\
0 &1 &1 &1 \\
1 &0 &1 &2 \\
2 &1 &0 &2 
\end{array} \right] 
\\
X^T = &\left[ \begin{array}{cccc}
2 &0 &1 &2 \\
2 &1 &0 &1 \\
1 &1 &1 &0 \\
0 &1 &2 &2 
\end{array} \right] ~~~~~~
Y = \left[ \begin{array}{cccc}
1 &0 &1 &2 \\
1 &2 &0 &1 \\
1 &1 &2 &0 \\
0 &1 &1 &1 
\end{array} \right]
\end{align*}

\begin{align*}
S(Y^T|X) = &\left[ \begin{array}{c|c|c|c|c|c|c|c|c}
0~0~0~0   &2   &1~1~1~0   &2   &1~0~2~1   &1   &0~0~0~0   &0   &0~0~0~0  \\
0~0~0~0   &0   &0~0~0~0   &1   &0~1~2~2   &1   &2~1~0~1   &1   &0~0~0~0  \\
1~1~1~0   &1   &0~0~0~0   &0   &0~0~0~0   &1   &2~0~1~2   &2   &1~1~1~0  \\
2~1~0~1   &2   &0~2~1~1   &1   &0~0~0~0   &0   &0~0~0~0   &2   &2~1~0~1    
\end{array} \right]
\\
S(X^T|Y) = &\left[ \begin{array}{c|c|c|c|c|c|c|c|c}
1~0~2~1   &1   &0~0~0~0   &0   &0~0~0~0   &1   &1~1~1~0   &1   &1~0~2~1  \\
0~1~2~2   &1   &2~1~0~1   &2   &0~0~0~0   &0   &0~0~0~0   &1   &0~1~2~2  \\
0~0~0~0   &1   &2~0~1~2   &1   &1~1~1~0   &2   &0~0~0~0   &0   &0~0~0~0  \\
0~0~0~0   &0   &0~0~0~0   &1   &2~1~0~1   &1   &0~2~1~1   &1   &0~0~0~0    
\end{array} \right]
\end{align*}
\end{example}

\subsection{Decompositions of the characteristic matrix}

In general, the matrices $X$ and $Y_1$ have minors of maximal full rank in the positions $I_0 \times I_1$. And the matrices $Y$ and $X_1$ have minors of maximal full rank in the positions $J_0 \times J_1$. From Lemma \ref{L:DE}, $X = X_0 Y_1^T X_0$ and $Y = Y_0 X_1^T Y_0$. With $X_0$ and $Y_0$  triangular and invertible, we recognize the format of a Bruhat decomposition (Section \ref{S:bruh}): Submatrices of $X$ and $Y_1^T$ of the same size and positioned in the south west corner are of same rank, and submatrices of $Y$ and $X_1^T$ of the same size and positioned in the north east corner are of same rank.

For reduced characteristic matrices and for the characteristic matrices of Theorem \ref{T:sred}, $X_1$ and $Y_1$ are zero outside the invertible minors. This is not the case in general. In general $Y_1$ can be nonzero outside the columns $I_1$ and the matrix $X_1$ can be nonzero outside the columns $J_1$. 

\begin{example}
\[
G = H =  \left[ \begin{array}{cccc}
1 &1 &0 &0 \\
0 &0 &1 &1
\end{array} \right]
\]
\[
X = \left[ \begin{array}{cccc}
1 &1 &0 &0 \\
1 &1 &1 &1 \\
0 &0 &1 &1 \\
1 &1 &1 &1 
\end{array} \right] = 
\left[ \begin{array}{cccc}
1 &1 &0 &0 \\
\dt &1 &1 &1 \\
\dt &\dt &1 &1 \\
\dt &\dt &\dt &1 
\end{array} \right] +
\left[ \begin{array}{cccc}
\dt &\dt &\dt &\dt \\
1 &\dt &\dt &\dt \\
0 &0 &\dt &\dt \\
1 &1 &1 &\dt 
\end{array} \right]
\qquad 
\begin{array}{l}
J_0 = \{ 2, 4 \} \\
J_1 = \{ 1, 3 \} \\
\end{array} 
\]
\[
Y = \left[ \begin{array}{cccc}
1 &1 &1 &1 \\
1 &1 &0 &0 \\
1 &1 &1 &1 \\
0 &0 &1 &1 
\end{array} \right] = 
\left[ \begin{array}{cccc}
\dt &1 &1 &1 \\
\dt &\dt &0 &0 \\
\dt &\dt &\dt &1 \\
\dt &\dt &\dt &\dt 
\end{array} \right] +
\left[ \begin{array}{cccc}
1 &\dt &\dt &\dt \\
1 &1 &\dt &\dt \\
1 &1 &1 &\dt \\
0 &0 &1 &1 
\end{array} \right]
\qquad 
\begin{array}{l}
I_0 = \{ 1, 3 \} \\
I_1 = \{ 2, 4 \} \\
\end{array} 
\]
\end{example}

In general 
\begin{align*}
&X = X [ n \times I_1 ] \times X[I_0 \times I_1]^{-1} \times X[ I_0 \times n ], \\
&Y = Y [ n \times J_1 ] \times Y[J_0 \times J_1]^{-1} \times Y[ J_0 \times n ].
\end{align*}
Let $W^T = X [ n \times I_1 ],$ $G = X[ I_0 \times n ],$ $V^T = Y [ n \times J_1 ],$ and $H = Y[ J_0 \times n ].$
In fact, as submatrices of $X$ and $Y$ containing the conventional generators, $G=G_{01}$ and $H=H_{10}$ are in reduced minimal span form.
The effect of the invertible minor is to bring $W^T$ and $V^T$ in systematic form in the positions $I_0$ and $J_0$.
Thus $X = W_0^T G_{01}$ and $Y = V_1^T H_{10}$.

\bigskip

In the special case where $X_1$ and $Y_1$ are zero outside their full rank minors, we have
\[
Y_1 = D_0 Y_1 D_1, ~~~~~X_1 = E_0 X_1 E_1,
\]
with $D_1$ and $E_1$ diagonal $0,1,$-matrices with support in $I_1$ and $J_1$, respectively. This is the case for reduced characteristic matrices and for the characteristic matrices in Theorem \ref{T:sred}. From $I = X_0 Y_1^T + X_1 Y_0^T = X Y_1^T + X_1 Y^T$,
$D_0^2  = (D_0 X D_1)(D_0 Y_1 D_1)^T$ and thus
\[
X[I_0 \times I_1] \times Y_1[I_0 \times I_1]^T = I.
\] 

\subsection{The transpose of a characteristic matrix} 

By Proposition \ref{P:abcd} the row span of a characteristic matrix $X$ is $|J_1|n$. 
 
\begin{lemma} Let $X$ be a left-ordered characteristic matrix with spanlength $|J_1| n$. The rows in $X^T$ have spans ending on the diagonal. The spanlength of $X^T$ is at least $|J_1|n$, with equality if and only if $X$ is reduced. 
Let $Y$ be a right-ordered characteristic matrix with spanlength $|I_1| n$. The rows in $Y^T$ have spans starting on the diagonal. The spanlength of $Y^T$ is at least $|I_1|n$, with equality if and only if $Y$ is reduced.   
\end{lemma}
\begin{proof}
Let $(i_0,i_1]$ be the span of a row $i_0$ in $X$. Then the span of row $i_1$ in $X^T$ includes $(i_0,i_1]$. The span is equal to $(i_0,i_1]$ for all $i_0$ if and only if $X$ is reduced. 
Similar for $Y$.
\end{proof}

\begin{theorem} \label{T:trans}
The transpose of a characteristic matrix is again a characteristic matrix if and only if the characteristic matrix is reduced. 
\begin{itemize}
\item[(X)] The transpose of a reduced left-ordered characteristic matrix $X$ is a right-ordered characteristic matrix $X^T$. A minimal span form for the row space of $X$ is given by
$X[I_0 \times n]$ and a minimal span form for the row space of $X^T$ by $X[n \times I_1]^T$. 
\item[(Y)] The transpose of a reduced right-ordered characteristic matrix $Y$ is a left-ordered characteristic matrix $Y^T$. A minimal span form for the row space of $Y$ is given by
$Y[J_0 \times n]$ and a minimal span form for the row space of $Y^T$ by $Y[n \times J_1]^T$.
\end{itemize} 
\end{theorem}
\begin{proof}
The lemma shows that the spanlength of a characteristic matrix increases when passing to the transpose unless the matrix is reduced. This proves the only if part. When the matrix is reduced taking the transpose preserves the set of characteristic spans and the spanlength, and the transpose is again a characteristic matrix. 
\end{proof}

To a reduced pair of dual characteristic matrices $X$ and $Y$ (with row spaces $G \perp H$) thus corresponds a reduced pair of dual characteristic matrices $Y^T$ and $X^T$ (with row spaces $V \perp W$). The matrices $G$ and $W$ are of the same size and share the same characteristic spans. The same for the matrices $H$ and $V$. 

\begin{example} There are three binary row spaces with characteristic spans $(0,3], (1,0], (2,4],$ $(3,1], (4,2]$. They are
\[
G = \left[ \begin{array}{ccccc} 1 &1 &0 &1 &0 \\ 0 &0 &1 &0 &1 \end{array} \right], ~~~
G' = \left[ \begin{array}{ccccc} 1 &1 &0 &1 &0 \\ 0 &0 &1 &1 &1 \end{array} \right], ~~~
G'' = \left[ \begin{array}{ccccc} 1 &1 &1 &1 &0 \\ 0 &0 &1 &0 &1 \end{array} \right].
\]
By the above $\{ X^T(G), X^T(G'), X^T(G'') \} = \{ Y(G), Y(G'), Y(G'') \}.$ In fact, $X^T(G) = Y(G),$ $X^T(G') = Y(G''), X^T(G'') = Y(G').$ Thus the matrices $G'$ and $G''$ describe the column spaces of each others reduced characteristic matrices.
\end{example}

Let $X$ have row space generated by $G$ and column space generated by $W^T$. The construction of a KV-trellis from a characteristic matrix requires a choice of independent generators from the rows of $X$. A set of rows in $X$ is independent if and only if the corresponding columns in $W$ are independent. 
Thus the set of possible choices is in bijection with the set of independent full minors of $W$, that is with the bases for the matroid of $W$. 

\begin{example} The only matrix in the previous example that admits a minimal trellis with $|V_i| \leq 2$ for all $i \in \{0,1,2,3,4 \}$ is $G'$. Namely such a trellis can only be formed with generators having spans $(2,4]$ and $(4,2]$. Those generators are independent if and only if, for $X^T(G) = Y(W)$, the columns $2$ and $4$ in $W$ are linearly independent. This occurs only if $W = G''$. The required trellis is generated by the rows $00111$ and $11101$.
\end{example}   

For reduced characteristic matrices $X$ and $Y$, the row spaces $G = X[I_0 \times n]$, $W = X[n \times I_1]^T$, $H = Y[J_0 \times n]$, and
$V = Y[n \times J_1]^T$ are all in minimal span form, with $G = G_{01},$ $W = W_{10},$ $H = H_{10}$, and $V = V_{01}.$ The following four matrices are in reduced minimal span form 
\[
\begin{array}{ccccc}
\left[ \begin{array}{c|c} X_0 &X_1 \\ \hline \noalign{\smallskip} 0 &G \end{array}\right]
&~~~
\left[ \begin{array}{c|c} W &0 \\ \hline \noalign{\smallskip} X_1^T &X_0^T \end{array}\right]
&~~~
\left[ \begin{array}{c|c} H &0 \\ \hline \noalign{\smallskip} Y_1 &Y_0 \end{array}\right]
&~~~
\left[ \begin{array}{c|c} Y_0^T &Y_1^T \\ \hline \noalign{\smallskip} 0 &V \end{array}\right]
\end{array}
\]
Moreover, for the systematic forms $G_1, W_0, H_0, V_1$,
\[
\begin{array}{ccccc}
X= W_0^T G_{01}
&~~
X^T = G_1^T W_{10}
&~~
Y= V_1^T H_{10}
&~~
Y^T = H_0^T V_{01}
\end{array}
\]

Circulant matrices are special cases of reduced characteristic matrices. In this case there is a clear relation between the row space and the column space of the circulant matrix. 

\begin{proposition}
Let $C'$ denote the reverse ordered version of a cyclic code $C$. Then $X(C)^T = Y(C')$ and $Y(C)^T = X(C')$.
\end{proposition}



\appendix

\section{A dual pair of reduced characteristic matrices} \label{S:binex}

Reduced characteristic matrices $X$ and $Y$ for a pair of binary orthogonal matrices $G$ and $H$.

{\small
\[
\begin{array}{ccm{2.75in}l}
\left[ \begin{array}{c|c} X_0 &X_1 \\ \hline 0 &G \end{array}\right] 
&=
&\begin{tikzpicture}
	[scale=.5,auto=center]
	\matrix [matrix of math nodes,ampersand replacement=\&,
                    left delimiter={[},right delimiter={]},row sep = .1em,column sep = .2em,text width=.7em,text height=.7em] (m)
        {
		1\&\;1\&\;1\&\;0\&\;0\&\;\zz\&\;\zz\&\;\zz\&\;\zz\&\;\zz\\ 
		\zz\&\;1\&\;1\&\;0\&\;1\&\;0\&\;\zz\&\;\zz\&\;\zz\&\;\zz\\ 
		\zz\&\;\zz\&\;1\&\;1\&\;0\&\;0\&\;0\&\;\zz\&\;\zz\&\;\zz\\ 
		\zz\&\;\zz\&\;\zz\&\;1\&\;1\&\;0\&\;1\&\;0\&\;\zz\&\;\zz\\ 
		\zz\&\;\zz\&\;\zz\&\;\zz\&\;1\&\;1\&\;0\&\;0\&\;0\&\;\zz\\
		\zz\&\;\zz\&\;\zz\&\;\zz\&\;\zz\&\;1\&\;1\&\;1\&\;0\&\;0 \\ 
		\zz\&\;\zz\&\;\zz\&\;\zz\&\;\zz\&\;\zz\&\;1\&\;1\&\;0\&\;1 \\ 
		\zz\&\;\zz\&\;\zz\&\;\zz\&\;\zz\&\;\zz\&\;\zz\&\;1\&\;1\&\;0  \\
        };  
        \draw[color=black] (m-1-6.north west) -- (m-8-6.south west);
        \draw[color=black] (m-5-1.south west) -- (m-5-10.south east);
        \draw[color=black,dashed] (m-4-4.north west) -- (m-4-6.north west) -- (m-5-6.south west) -- (m-5-4.south west) -- (m-4-4.north west);
        \draw[color=black,dashed] (m-4-6.north west) -- (m-4-8.north west) -- (m-5-8.south west) -- (m-5-6.south west) -- (m-4-6.north west);
        \draw[color=black] (m-1-1.north west) -- (m-1-4.north west) -- (m-4-4.north west) -- (m-4-1.north west) -- (m-1-1.north west);
        \draw[color=black] (m-5-8.south west) -- (m-5-10.south east) -- (m-8-10.south east) -- (m-8-8.south west) -- (m-6-8.north west);
\end{tikzpicture}
& \\  \noalign{\smallskip}
&\perp
&\begin{tikzpicture}
	[scale=.5,auto=center]
	\matrix [matrix of math nodes,ampersand replacement=\&,
                    left delimiter={[},right delimiter={]},row sep = .1em,column sep = .2em,text width=.7em,text height=.7em] (m)
        {
		0\&\;1\&\;1\&\;1\&\;0\&\;0\&\;1\&\;1\&\;1\&\;\zz\\
		1\&\;0\&\;1\&\;1\&\;1\&\;1\&\;0\&\;1\&\;1\&\;1\\
        };  
        \draw[color=black] (m-1-6.north west) -- (m-2-6.south west);
        \draw[color=black] (m-1-4.north west) -- (m-1-6.north west) -- (m-2-6.south west) -- (m-2-4.south west) -- (m-1-4.north west);
        \draw[color=black] (m-1-6.north west) -- (m-1-8.north west) -- (m-2-8.south west) -- (m-2-6.south west) -- (m-1-6.north west);
\end{tikzpicture}
&=~ \left[ ~H~|~H~ \right]. \\ \noalign{\bigskip}
\left[ ~G~|~G~ \right]
&=
&\begin{tikzpicture}
	[scale=.5,auto=center]
	\matrix [matrix of math nodes,ampersand replacement=\&,
                    left delimiter={[},right delimiter={]},row sep = .1em,column sep = .2em,text width=.7em,text height=.7em] (m)
        {
		1\&\;1\&\;1\&\;0\&\;0\&\;1\&\;1\&\;1\&\;\zz\&\;\zz\\ 
		\zz\&\;1\&\;1\&\;0\&\;1\&\;0\&\;1\&\;1\&\;0\&\;1\\ 
		\zz\&\;\zz\&\;1\&\;1\&\;0\&\;0\&\;0\&\;1\&\;1\&\;0\\ 
        };  
        \draw[color=black] (m-1-6.north west) -- (m-3-6.south west);
        \draw[color=black] (m-1-3.north west) -- (m-1-6.north west) -- (m-3-6.south west) -- (m-3-3.south west) -- (m-1-3.north west);
        \draw[color=black] (m-1-6.north west) -- (m-1-9.north west) -- (m-3-9.south west) -- (m-3-6.south west) -- (m-1-6.north west);
\end{tikzpicture} 
& \\  \noalign{\smallskip}
&\perp
&\begin{tikzpicture}
	[scale=.5,auto=center]
	\matrix [matrix of math nodes,ampersand replacement=\&,
                    left delimiter={[},right delimiter={]},row sep = .1em,column sep = .2em,text width=.7em,text height=.7em] (m)
        {
		0\&\;1\&\;1\&\;1\&\;\zz\&\;\zz\&\;\zz\&\;\zz\&\;\zz\&\;\zz\\
		1\&\;0\&\;1\&\;1\&\;1\&\;\zz\&\;\zz\&\;\zz\&\;\zz\&\;\zz\\\hline
		\zz\&\;0\&\;1\&\;1\&\;1\&\;1\&\;\zz\&\;\zz\&\;\zz\&\;\zz\\ 
		\zz\&\;\zz\&\;0\&\;0\&\;1\&\;1\&\;1\&\;\zz\&\;\zz\&\;\zz\\ 
		\zz\&\;\zz\&\;\zz\&\;1\&\;0\&\;0\&\;1\&\;1\&\;\zz\&\;\zz\\ 
		\zz\&\;\zz\&\;\zz\&\;\zz\&\;0\&\;0\&\;1\&\;1\&\;1\&\;\zz\\ 
		\zz\&\;\zz\&\;\zz\&\;\zz\&\;\zz\&\;1\&\;0\&\;1\&\;1\&\;1 \\
        };  
        \draw[color=black] (m-1-6.north west) -- (m-7-6.south west);
        \draw[color=black] (m-3-1.north west) -- (m-3-10.north east);
        \draw[color=black,dashed] (m-3-3.north west) -- (m-3-6.north west) -- (m-5-6.south west) -- (m-5-3.south west) -- (m-3-3.north west);
        \draw[color=black,dashed] (m-3-6.north west) -- (m-3-9.north west) -- (m-5-9.south west) -- (m-5-6.south west) -- (m-3-6.north west);
        \draw[color=black] (m-1-1.north west) -- (m-1-3.north west) -- (m-3-3.north west) -- (m-3-1.north west) -- (m-1-1.north west);
        \draw[color=black] (m-5-9.south west) -- (m-5-10.south east) -- (m-7-10.south east) -- (m-7-9.south west) -- (m-5-9.south west);
\end{tikzpicture}
&=~ \left[ \begin{array}{c|c} H &0 \\ \hline \noalign{\smallskip} Y_1 &Y_0 \end{array}\right]
\end{array}
\]
}
\bigskip
{\small
\[
X = 
\left[ \begin {array}{ccccc} 
1&1&1&0&0\\ 
\dt&1&1&0&1\\ 
\dt&\dt&1&1&0\\ 
\dt&\dt&\dt&1&1\\ 
\dt&\dt&\dt&\dt&1
\end {array} \right] +
\left[ \begin {array}{ccccc} 
\dt&\dt&\dt&\dt&\dt\\ 
0&\dt&\dt&\dt&\dt\\ 
0&0&\dt&\dt&\dt\\ 
0&1&0&\dt&\dt\\ 
1&0&0&0&\dt
\end {array} \right] 
=
\left[ \begin {array}{ccccc} 
1&1&1&0&0\\ 
0&1&1&0&1\\ 
0&0&1&1&0\\ 
0&1&0&1&1\\ 
1&0&0&0&1\end {array} \right]
~~\begin {array}{ccccc} 
(0,2] \\ 
(1,4] \\ 
(2,3] \\ 
(3,1] \\ 
(4,0] \end {array}
\]
\medskip
\[
Y  =
\left[ \begin {array}{ccccc} 
\dt&0&1&1&1\\ 
\dt&\dt&0&0&1\\ 
\dt&\dt&\dt&1&0\\ 
\dt&\dt&\dt&\dt&0\\ 
\dt&\dt&\dt&\dt&\dt\end {array} \right] +
\left[ \begin {array}{ccccc} 
1&\dt&\dt&\dt&\dt\\ 
1&1&\dt&\dt&\dt\\ 
0&1&1&\dt&\dt\\ 
0&1&1&1&\dt \\ 
1&0&1&1&1\end {array} \right]
= 
\left[ \begin {array}{ccccc} 
1&0&1&1&1\\ 
1&1&0&0&1\\ 
0&1&1&1&0\\ 
0&1&1&1&0\\ 
1&0&1&1&1\end {array} \right]
~~\begin {array}{ccccc} 
(2,0] \\ 
(4,1] \\ 
(3,2] \\ 
(1,3] \\ 
(0,4] \end {array}
\]
}

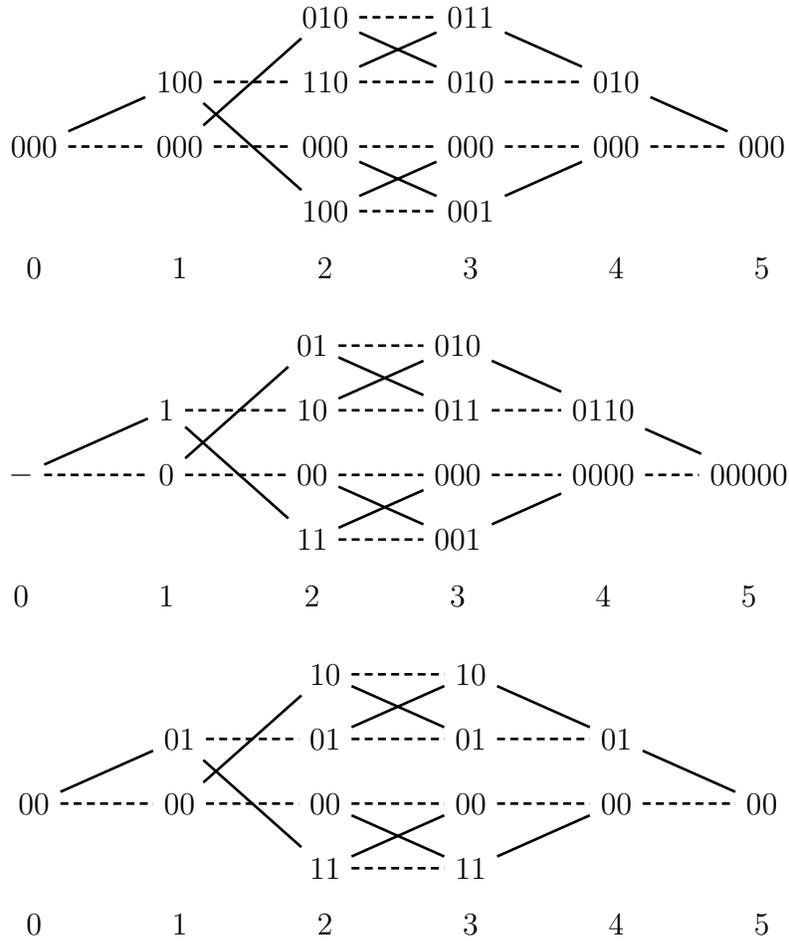
\begin{figure} \label{F:conv3}
{
\begin{center}
\begin{tikzpicture}
[scale=.43,auto=center,every node/.style={minimum size=0em}]

\node[right] at (21,10) {$G = \left( \begin{array}{ccccc} 1 &1 &1 &0 &0 \\ 0 &1 &1 &0 &1 \\ 0 &0 &1 &1 &0 \end{array} \right)$};
\node[left] at (43.5,10) {$H = \left( \begin{array}{ccccc}  0 &1 &1 &1 &0  \\ 1 &0 &1 &1 &1  \end{array} \right)$};
\node at (32.25,10.2) {$\perp$};

\newcommand{\y}{-2.75}
\newcommand{\x}{2}
\newcommand{\z}{16.5}
\newcommand{\m}{1.5}
\node[] at (\z+3*\m,\y) {0};
\node[] at (\z+6*\m,\y) {1};
\node[] at (\z+9*\m,\y) {2};
\node[] at (\z+12*\m,\y) {3};
\node[] at (\z+15*\m,\y) {4};
\node[] at (\z+18*\m,\y) {5};



\node (s0) at (\z+3*\m,1) {$000$};
\node (s1) at (\z+6*\m,1) {$000$};
\node (t1) at (\z+6*\m,3) {$100$};
\node (s2) at (\z+9*\m,1) {$000$};
\node (t2) at (\z+9*\m,5) {$010$};
\node (u2) at (\z+9*\m,-1) {$100$};
\node (v2) at (\z+9*\m,3) {$110$};
\node (s3) at (\z+12*\m,1) {$000$};
\node (t3) at (\z+12*\m,-1) {$001$};
\node (u3) at (\z+12*\m,5) {$011$};
\node (v3) at (\z+12*\m,3) {$010$};
\node (s4) at (\z+15*\m,1) {$000$};
\node (t4) at (\z+15*\m,3) {$010$};
\node  (s5) at (\z+18*\m,1) {$000$};

 \path
(s0) edge [-,line width=1pt,densely dashed,color=black] (s1)
(s0) edge [-,line width=1pt,color=black] (t1)
(s1) edge [-,line width=1pt,densely dashed,color=black] (s2)
(s1) edge [-,line width=1pt,color=black] (t2)
(t1) edge [-,line width=1pt,densely dashed,color=black] (v2)
(t1) edge [-,line width=1pt,color=black] (u2)
(s2) edge [-,line width=1pt,densely dashed,color=black] (s3)
(s2) edge [-,line width=1pt,color=black] (t3)
(v2) edge [-,line width=1pt,densely dashed,color=black] (v3)
(v2) edge [-,line width=1pt,color=black] (u3)
(u2) edge [-,line width=1pt,color=black] (s3)
(u2) edge [-,line width=1pt,densely dashed,color=black] (t3)
(t2) edge [-,line width=1pt,color=black] (v3)
(t2) edge [-,line width=1pt,densely dashed,color=black] (u3)
(s3) edge [-,line width=1pt,densely dashed,color=black] (s4)
(u3) edge [-,line width=1pt,color=black] (t4)
(v3) edge [-,line width=1pt,densely dashed,color=black] (t4)
(t3) edge [-,line width=1pt,color=black] (s4)
(t4) edge [-,line width=1pt,color=black] (s5)
(s4) edge [-,line width=1pt,densely dashed,color=black] (s5);

\end{tikzpicture}
\end{center}

\begin{center}
\begin{tikzpicture}
[scale=.43,auto=center,every node/.style={minimum size=0em}]


\newcommand{\y}{-2.75}
\newcommand{\x}{2}
\newcommand{\z}{16.5}
\newcommand{\m}{1.5}
\node[] at (\z+3*\m,\y) {0};
\node[] at (\z+6*\m,\y) {1};
\node[] at (\z+9*\m,\y) {2};
\node[] at (\z+12*\m,\y) {3};
\node[] at (\z+15*\m,\y) {4};
\node[] at (\z+18*\m,\y) {5};



\node (s0) at (\z+3*\m,1) {$-$};
\node (s1) at (\z+6*\m,1) {$0$};
\node (t1) at (\z+6*\m,3) {$1$};
\node (s2) at (\z+9*\m,1) {$00$};
\node (t2) at (\z+9*\m,5) {$01$};
\node (u2) at (\z+9*\m,-1) {$11$};
\node (v2) at (\z+9*\m,3) {$10$};
\node (s3) at (\z+12*\m,1) {$000$};
\node (t3) at (\z+12*\m,-1) {$001$};
\node (u3) at (\z+12*\m,5) {$010$};
\node (v3) at (\z+12*\m,3) {$011$};
\node (s4) at (\z+15*\m,1) {$0000$};
\node (t4) at (\z+15*\m,3) {$0110$};
\node  (s5) at (\z+18*\m,1) {$00000$};

 \path
(s0) edge [-,line width=1pt,densely dashed,color=black] (s1)
(s0) edge [-,line width=1pt,color=black] (t1)
(s1) edge [-,line width=1pt,densely dashed,color=black] (s2)
(s1) edge [-,line width=1pt,color=black] (t2)
(t1) edge [-,line width=1pt,densely dashed,color=black] (v2)
(t1) edge [-,line width=1pt,color=black] (u2)
(s2) edge [-,line width=1pt,densely dashed,color=black] (s3)
(s2) edge [-,line width=1pt,color=black] (t3)
(v2) edge [-,line width=1pt,densely dashed,color=black] (v3)
(v2) edge [-,line width=1pt,color=black] (u3)
(u2) edge [-,line width=1pt,color=black] (s3)
(u2) edge [-,line width=1pt,densely dashed,color=black] (t3)
(t2) edge [-,line width=1pt,color=black] (v3)
(t2) edge [-,line width=1pt,densely dashed,color=black] (u3)
(s3) edge [-,line width=1pt,densely dashed,color=black] (s4)
(u3) edge [-,line width=1pt,color=black] (t4)
(v3) edge [-,line width=1pt,densely dashed,color=black] (t4)
(t3) edge [-,line width=1pt,color=black] (s4)
(t4) edge [-,line width=1pt,color=black] (s5)
(s4) edge [-,line width=1pt,densely dashed,color=black] (s5);

\end{tikzpicture}
\end{center}

\begin{center}
\begin{tikzpicture}
[scale=.43,auto=center,every node/.style={minimum size=0em}]


\newcommand{\y}{-2.75}
\newcommand{\x}{2}
\newcommand{\z}{16.5}
\newcommand{\m}{1.5}
\node[] at (\z+3*\m,\y) {0};
\node[] at (\z+6*\m,\y) {1};
\node[] at (\z+9*\m,\y) {2};
\node[] at (\z+12*\m,\y) {3};
\node[] at (\z+15*\m,\y) {4};
\node[] at (\z+18*\m,\y) {5};



\node (s0) at (\z+3*\m,1) {$00$};
\node (s1) at (\z+6*\m,1) {$00$};
\node (t1) at (\z+6*\m,3) {$01$};
\node (s2) at (\z+9*\m,1) {$00$};
\node (t2) at (\z+9*\m,5) {$10$};
\node (u2) at (\z+9*\m,-1) {$11$};
\node (v2) at (\z+9*\m,3) {$01$};
\node (s3) at (\z+12*\m,1) {$00$};
\node (t3) at (\z+12*\m,-1) {$11$};
\node (u3) at (\z+12*\m,5) {$10$};
\node (v3) at (\z+12*\m,3) {$01$};
\node (s4) at (\z+15*\m,1) {$00$};
\node (t4) at (\z+15*\m,3) {$01$};
\node  (s5) at (\z+18*\m,1) {$00$};

 \path
(s0) edge [-,line width=1pt,densely dashed,color=black] (s1)
(s0) edge [-,line width=1pt,color=black] (t1)
(s1) edge [-,line width=1pt,densely dashed,color=black] (s2)
(s1) edge [-,line width=1pt,color=black] (t2)
(t1) edge [-,line width=1pt,densely dashed,color=black] (v2)
(t1) edge [-,line width=1pt,color=black] (u2)
(s2) edge [-,line width=1pt,densely dashed,color=black] (s3)
(s2) edge [-,line width=1pt,color=black] (t3)
(v2) edge [-,line width=1pt,densely dashed,color=black] (v3)
(v2) edge [-,line width=1pt,color=black] (u3)
(u2) edge [-,line width=1pt,color=black] (s3)
(u2) edge [-,line width=1pt,densely dashed,color=black] (t3)
(t2) edge [-,line width=1pt,color=black] (v3)
(t2) edge [-,line width=1pt,densely dashed,color=black] (u3)
(s3) edge [-,line width=1pt,densely dashed,color=black] (s4)
(u3) edge [-,line width=1pt,color=black] (t4)
(v3) edge [-,line width=1pt,densely dashed,color=black] (t4)
(t3) edge [-,line width=1pt,color=black] (s4)
(t4) edge [-,line width=1pt,color=black] (s5)
(s4) edge [-,line width=1pt,densely dashed,color=black] (s5);

\end{tikzpicture}
\end{center}
\caption{Three different labelings for the minimal conventional trellises of the row space of $G$: vertex labeling by information symbols (top), by codeword symbols (middle), and by syndromes (bottom)} 
} 
\end{figure}


\begin{figure}
\begin{center}
\begin{tikzpicture}
[scale=.43,auto=center,every node/.style={minimum size=0em}]

\node[left] at (14.5,2) {$G = \left( \begin{array}{ccccc} 1 &1 &1 &0 &0 \\ 0 &1 &1 &0 &1 \\ 0 &0 &1 &1 &0 \end{array} \right)$};

\newcommand{\y}{-2.75}
\newcommand{\x}{2}
\newcommand{\z}{16.5}

\newcommand{\w}{6.75}
\node[] at (\z+3+1.5,\w) {01};
\node[] at (\z+6+1.5,\w) {10};
\node[] at (\z+9+1.5,\w) {11};
\node[] at (\z+12+1.5,\w) {11};
\node[] at (\z+15+1.5,\w) {01};

\node[left] at (\z+\x,-1){11};
\node[left] at (\z+\x,1) {00};
\node[left] at (\z+\x,3) {01};
\node[left] at (\z+\x,5) {10};

\node (s0) at (\z+3,1) {$\bullet$};
\node (s1) at (\z+6,1) {$\bullet$};
\node (t1) at (\z+6,3) {$\bullet$};
\node (s2) at (\z+9,1) {$\bullet$};
\node (t2) at (\z+9,5) {$\bullet$};
\node (u2) at (\z+9,-1) {$\bullet$};
\node (v2) at (\z+9,3) {$\bullet$};
\node (s3) at (\z+12,1) {$\bullet$};
\node (t3) at (\z+12,-1) {$\bullet$};
\node (u3) at (\z+12,5) {$\bullet$};
\node (v3) at (\z+12,3) {$\bullet$};
\node (s4) at (\z+15,1) {$\bullet$};
\node (t4) at (\z+15,3) {$\bullet$};
\node  (s5) at (\z+18,1) {$\bullet$};

 \path
(s0) edge [-,line width=1pt,densely dashed,color=black] (s1)
(s0) edge [-,line width=1pt,color=black] (t1)
(s1) edge [-,line width=1pt,densely dashed,color=black] (s2)
(s1) edge [-,line width=1pt,color=black] (t2)
(t1) edge [-,line width=1pt,densely dashed,color=black] (v2)
(t1) edge [-,line width=1pt,color=black] (u2)
(s2) edge [-,line width=1pt,densely dashed,color=black] (s3)
(s2) edge [-,line width=1pt,color=black] (t3)
(v2) edge [-,line width=1pt,densely dashed,color=black] (v3)
(v2) edge [-,line width=1pt,color=black] (u3)
(u2) edge [-,line width=1pt,color=black] (s3)
(u2) edge [-,line width=1pt,densely dashed,color=black] (t3)
(t2) edge [-,line width=1pt,color=black] (v3)
(t2) edge [-,line width=1pt,densely dashed,color=black] (u3)
(s3) edge [-,line width=1pt,densely dashed,color=black] (s4)
(u3) edge [-,line width=1pt,color=black] (t4)
(v3) edge [-,line width=1pt,densely dashed,color=black] (t4)
(t3) edge [-,line width=1pt,color=black] (s4)
(t4) edge [-,line width=1pt,color=black] (s5)
(s4) edge [-,line width=1pt,densely dashed,color=black] (s5);


\renewcommand{\y}{-8.5}
\renewcommand{\w}{-2.75}

\node[left] at (14.5,2+\y) {$H = \left( \begin{array}{ccccc}  0 &1 &1 &1 &0  \\ 1 &0 &1 &1 &1  \end{array} \right)$};

\node[] at (\z+3+1.5,\w+\y) {100};
\node[] at (\z+6+1.5,\w+\y) {110};
\node[] at (\z+9+1.5,\w+\y) {111};
\node[] at (\z+12+1.5,\w+\y) {001};
\node[] at (\z+15+1.5,\w+\y) {010};

\node[right] at (\z+\x+17,-1+\y) {110};
\node[right] at (\z+\x+17,0+\y) {001};
\node[right] at (\z+\x+17,1+\y) {000};

\node[right] at (\z+\x+17,3+\y) {100};
\node[right] at (\z+\x+17,4+\y) {011};
\node[right] at (\z+\x+17,5+\y) {010};

\node (s0) at (\z+3,1+\y) {$\bullet$};

\node (s1) at (\z+6,1+\y) {$\bullet$};
\node (t1) at (\z+6,3+\y) {$\bullet$};

\node (s2) at (\z+9,1+\y) {$\bullet$};
\node (t2) at (\z+9,-1+\y) {$\bullet$};
\node (u2) at (\z+9,3+\y) {$\bullet$};
\node (v2) at (\z+9,5+\y) {$\bullet$};

\node (s3) at (\z+12,1+\y) {$\bullet$};
\node (t3) at (\z+12,0+\y) {$\bullet$};
\node (u3) at (\z+12,4+\y) {$\bullet$};
\node (v3) at (\z+12,5+\y) {$\bullet$};

\node (s4) at (\z+15,1+\y) {$\bullet$};
\node (t4) at (\z+15,5+\y) {$\bullet$};

\node  (s5) at (\z+18,1+\y) {$\bullet$};

 \path
(s0) edge [-,line width=1pt,densely dashed,color=black] (s1)
(s0) edge [-,line width=1pt,color=black] (t1)
(s1) edge [-,line width=1pt,densely dashed,color=black] (s2)
(s1) edge [-,line width=1pt,color=black] (t2)
(t1) edge [-,line width=1pt,densely dashed,color=black] (u2)
(t1) edge [-,line width=1pt,color=black] (v2)
(s2) edge [-,line width=1pt,densely dashed,color=black] (s3)
(v2) edge [-,line width=1pt,densely dashed,color=black] (v3)
(u2) edge [-,line width=1pt,color=black] (u3)
(t2) edge [-,line width=1pt,color=black] (t3)
(s3) edge [-,line width=1pt,densely dashed,color=black] (s4)
(u3) edge [-,line width=1pt,color=black] (t4)
(v3) edge [-,line width=1pt,densely dashed,color=black] (t4)
(t3) edge [-,line width=1pt,color=black] (s4)
(t4) edge [-,line width=1pt,color=black] (s5)
(s4) edge [-,line width=1pt,densely dashed,color=black] (s5);

\end{tikzpicture}
\end{center}
\caption{A pair of dual conventional trellises, based on rows $X_{1,2,3}$ and $Y_{4,5}$} \label{fig:tr}
\end{figure}
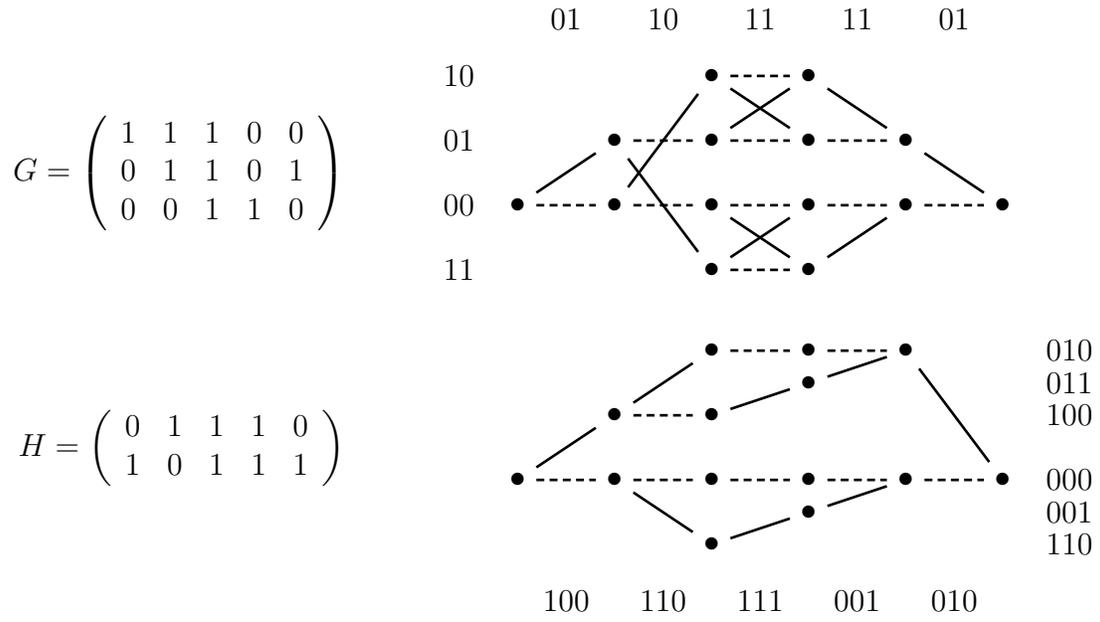

\begin{figure}
\begin{center}
\begin{tikzpicture}
[scale=.43,auto=center,every node/.style={minimum size=0em}]

\node[left] at (14.5,1) {$G = \left( \begin{array}{lcccr} 1 &1 &1 &0 &0 \\ 1] &0 &0 &0 &(1 \\ 0 &0 &1 &1 &0 \end{array} \right)$};

\newcommand{\y}{-2.75}
\newcommand{\x}{2}
\newcommand{\z}{16.5}

\newcommand{\w}{4.75}
\node[] at (\z+3+1.5,\w) {01};
\node[] at (\z+6+1.5,\w) {11};
\node[] at (\z+9+1.5,\w) {10};
\node[] at (\z+12+1.5,\w) {10};
\node[] at (\z+15+1.5,\w) {01};

\node[left] at (\z+\x,-1){10};
\node[left] at (\z+\x,1) {00};
\node[left] at (\z+\x,3) {01};

\renewcommand{\y}{0}

\node (s0) at (\z+3,1+\y) {$\bullet$};
\node (t0) at (\z+3,3+\y) {$\bullet$};

\node (s1) at (\z+6,1+\y) {$\bullet]$};
\node (t1) at (\z+6,3+\y) {$\bullet$};

\node (s2) at (\z+9,1+\y) {$\bullet$};
\node (u2) at (\z+9,-1+\y) {$\bullet$};

\node (s3) at (\z+12,1+\y) {$\bullet$};
\node (t3) at (\z+12,-1+\y) {$\bullet$};

\node (s4) at (\z+15,1+\y) {$(\bullet$};

\node (s5) at (\z+18,1+\y) {$\bullet$};
\node (t5) at (\z+18,3+\y) {$\bullet$};
 
\path
(s0) edge [-,line width=1pt,densely dashed,color=black] (s1)
(s0) edge [-,line width=1pt,color=black] (t1)
(t0) edge [-,line width=1pt,densely dashed,color=black] (t1)
(t0) edge [-,line width=1pt,color=black] (s1)
(s1) edge [-,line width=1pt,densely dashed,color=black] (s2)
(t1) edge [-,line width=1pt,color=black] (u2)
(s2) edge [-,line width=1pt,densely dashed,color=black] (s3)
(s2) edge [-,line width=1pt,color=black] (t3)
(u2) edge [-,line width=1pt,color=black] (s3)
(u2) edge [-,line width=1pt,densely dashed,color=black] (t3)
(s3) edge [-,line width=1pt,densely dashed,color=black] (s4)
(t3) edge [-,line width=1pt,color=black] (s4)
(s4) edge [-,line width=1pt,densely dashed,color=black] (s5)
(s4) edge [-,line width=1pt,color=black] (t5);

\renewcommand{\y}{-6.5}
\renewcommand{\w}{-2.75}

\node[left] at (14.5,1+\y) {$H = \left( \begin{array}{clccr}  0 &1 &1 &1 &0  \\ 1 &1] &0 &0 &(1  \end{array} \right)$};

\node[] at (\z+3+1.5,\w+\y) {110};
\node[] at (\z+6+1.5,\w+\y) {100};
\node[] at (\z+9+1.5,\w+\y) {101};
\node[] at (\z+12+1.5,\w+\y) {001};
\node[] at (\z+15+1.5,\w+\y) {010};

\node[right] at (\z+\x+17,-1+\y) {100};
\node[right] at (\z+\x+17,0+\y) {001};
\node[right] at (\z+\x+17,1+\y) {000};

\node[right] at (\z+\x+17,3+\y) {010};

\node (s0) at (\z+3,3+\y) {$\bullet$};
\node (t0) at (\z+3,1+\y) {$\bullet$};

\node (s1) at (\z+6,-1+\y) {$\bullet$};
\node (t1) at (\z+6,1+\y) {$\bullet$};

\node (s2) at (\z+9,1+\y) {$\bullet]$};
\node (t2) at (\z+9,-1+\y) {$\bullet$};

\node (s3) at (\z+12,1+\y) {$\bullet$};
\node (t3) at (\z+12,0+\y) {$\bullet$};

\node (s4) at (\z+15,1+\y) {$(\bullet$};

\node  (s5) at (\z+18,3+\y) {$\bullet$};
\node  (t5) at (\z+18,1+\y) {$\bullet$};

 \path
(s0) edge [-,line width=1pt,color=black] (s1)
(t0) edge [-,line width=1pt,densely dashed,color=black] (t1)
(s1) edge [-,line width=1pt,densely dashed,color=black] (t2)
(s1) edge [-,line width=1pt,color=black] (s2)
(t1) edge [-,line width=1pt,densely dashed,color=black] (s2)
(t1) edge [-,line width=1pt,color=black] (t2)
(s2) edge [-,line width=1pt,densely dashed,color=black] (s3)
(t2) edge [-,line width=1pt,color=black] (t3)
(s3) edge [-,line width=1pt,densely dashed,color=black] (s4)
(t3) edge [-,line width=1pt,color=black] (s4)
(s4) edge [-,line width=1pt,color=black] (s5)
(s4) edge [-,line width=1pt,densely dashed,color=black] (t5);

\end{tikzpicture}
\end{center}
\caption{A pair of dual tailbiting trellises, based on rows $X_{1,5,3}$ and $Y_{4,2}$} \label{fig:tb}
\end{figure}
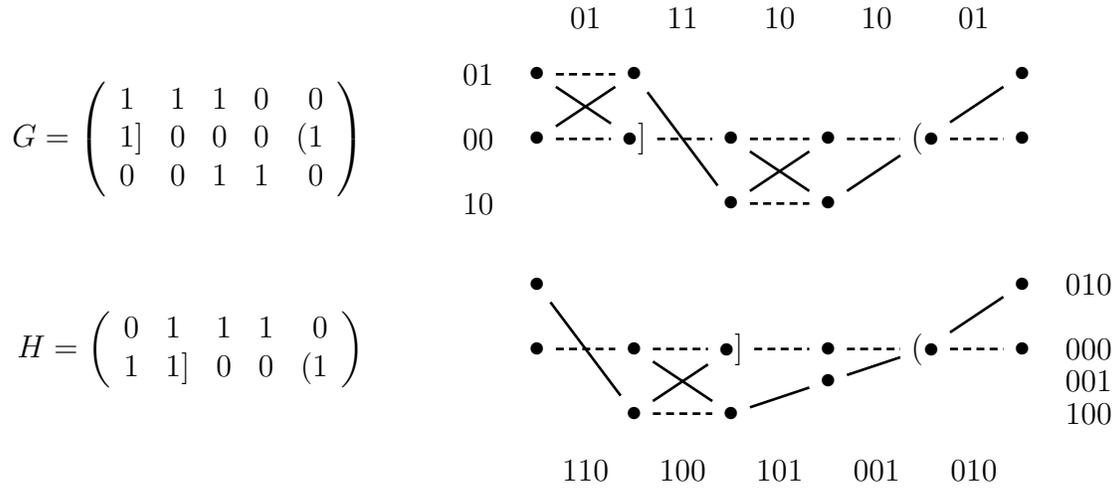

Label codes for KV-trellises constructed from $X$ and $Y$ using the span based BCJR construction. 

\begin{align*}
Y^T =
&\left[ \begin {array}{ccccc} 
1&1&0&0&1\\ 
0&1&1&1&0\\ 
1&0&1&1&1\\ 
1&0&1&1&1\\ 
1&1&0&0&1\end {array} \right]
\qquad
X=
 \left[ \begin {array}{ccccc} 1&1&1&0&0\\ 0&1&1&0&1
\\ 0&0&1&1&0\\ 0&1&0&1&1
\\ 1&0&0&0&1\end {array} \right]
\\
X^T = 
&\left[ \begin {array}{ccccc} 
1&0&0&0&1\\ 
1&1&0&1&0\\ 
1&1&1&0&0\\ 
0&0&1&1&0\\ 
0&1&0&1&1\end {array} \right]
\qquad
Y = 
 \left[ \begin {array}{ccccc} 1&0&1&1&1\\ 1&1&0&0&1
\\ 0&1&1&1&0\\ 0&1&1&1&0
\\ 1&0&1&1&1\end {array} \right]
\end{align*}

\begin{align*}
S(Y^T|X) = &\left[ \begin {array}{c|c|c|c|c|c|c|c|c|c|c} 
0~0~0~0~0   &1   &1~1~0~0~1   &1   &1~0~1~1~1   &1   &0~0~0~0~0   &0   &0~0~0~0~0   &0   &0~0~0~0~0   \\ 
0~0~0~0~0   &0   &0~0~0~0~0   &1   &0~1~1~1~0   &1   &1~1~0~0~1   &0   &1~1~0~0~1   &1   &0~0~0~0~0   \\ 
0~0~0~0~0   &0   &0~0~0~0~0   &0   &0~0~0~0~0   &1   &1~0~1~1~1   &1   &0~0~0~0~0   &0   &0~0~0~0~0   \\ 
0~1~1~1~0   &0   &0~1~1~1~0   &1   &0~0~0~0~0   &0   &0~0~0~0~0   &1   &1~0~1~1~1   &1   &0~1~1~1~0   \\ 
1~1~0~0~1   &1   &0~0~0~0~0   &0   &0~0~0~0~0   &0   &0~0~0~0~0   &0   &0~0~0~0~0   &1   &1~1~0~0~1   \end {array} \right]
\\ \noalign{\medskip}
S(X^T|Y) = &\left[ \begin {array}{c|c|c|c|c|c|c|c|c|c|c} 
1~0~0~0~1   &1   &0~0~0~0~0   &0   &0~0~0~0~0   &1   &1~1~1~0~0   &1   &1~1~0~1~0   &1   &1~0~0~0~1   \\ 
0~1~0~1~1   &1   &1~1~0~1~0   &1   &0~0~0~0~0   &0   &0~0~0~0~0   &0   &0~0~0~0~0   &1   &0~1~0~1~1   \\ 
0~0~1~1~0   &0   &0~0~1~1~0   &1   &1~1~1~0~0   &1   &0~0~0~0~0   &1   &0~0~1~1~0   &0   &0~0~1~1~0   \\ 
0~0~0~0~0   &0   &0~0~0~0~0   &1   &1~1~0~1~0   &1   &0~0~1~1~0   &1   &0~0~0~0~0   &0   &0~0~0~0~0   \\ 
0~0~0~0~0   &1   &1~0~0~0~1   &0   &1~0~0~0~1   &1   &0~1~1~0~1   &1   &0~1~0~1~1   &1   &0~0~0~0~0   \end {array} \right]
\end{align*}

\bigskip

Label codes for the dual trellises defined with rows $1,3,5$ of $X$ and rows $2,4$ of $Y$. 

\newcommand{\du}{{\cdot\,}}

\begin{align*}
S(Y_{2,4}^T|X_{1,3,5}) =  &\left[ \begin {array}{c|c|c|c|c|c|c|c|c|c|c} 
\,\du~0~\du~0~\du   &1   &\du~1~\du~0~\du   &1   &\du~0~\du~1~\du   &1   &\du~0~\du~0~\du   &0   &\du~0~\du~0~\du   &0   &\du~0~\du~0~\du\,   \\ 
\,\du~\du~\du~\du~\du   &\du   &\du~\du~\du~\du~\du   &\du   &\du~\du~\du~\du~\du   &\du   &\du~\du~\du~\du~\du   &\du   &\du~\du~\du~\du~\du   &\du   &\du~\du~\du~\du~\du\,   \\ 
\,\du~0~\du~0~\du   &0   &\du~0~\du~0~\du   &0   &\du~0~\du~0~\du   &1   &\du~0~\du~1~\du   &1   &\du~0~\du~0~\du   &0   &\du~0~\du~0~\du\,   \\ 
\,\du~\du~\du~\du~\du   &\du   &\du~\du~\du~\du~\du   &\du   &\du~\du~\du~\du~\du   &\du   &\du~\du~\du~\du~\du   &\du   &\du~\du~\du~\du~\du   &\du   &\du~\du~\du~\du~\du\,   \\ 
\,\du~1~\du~0~\du   &1   &\du~0~\du~0~\du   &0   &\du~0~\du~0~\du   &0   &\du~0~\du~0~\du   &0   &\du~0~\du~0~\du   &1   &\du~1~\du~0~\du\,   \end {array} \right] 
\\ \noalign{\medskip}
S(X_{1,3,5}^T|Y_{2,4}) = &\left[ \begin {array}{c|c|c|c|c|c|c|c|c|c|c} 
\du~\du~\du~\du~\du   &\du   &\du~\du~\du~\du~\du   &\du   &\du~\du~\du~\du~\du   &\du   &\du~\du~\du~\du~\du   &\du   &\du~\du~\du~\du~\du   &\du   &\du~\du~\du~\du~\du   \\ 
0~\du~0~\du~1   &1   &1~\du~0~\du~0   &1   &0~\du~0~\du~0   &0   &0~\du~0~\du~0   &0   &0~\du~0~\du~0   &1   &0~\du~0~\du~1   \\ 
\du~\du~\du~\du~\du   &\du   &\du~\du~\du~\du~\du   &\du   &\du~\du~\du~\du~\du   &\du   &\du~\du~\du~\du~\du   &\du   &\du~\du~\du~\du~\du   &\du   &\du~\du~\du~\du~\du   \\ 
0~\du~0~\du~0   &0   &0~\du~0~\du~0   &1   &1~\du~0~\du~0   &1   &0~\du~1~\du~0   &1   &0~\du~0~\du~0   &0   &0~\du~0~\du~0   \\ 
\du~\du~\du~\du~\du   &\du   &\du~\du~\du~\du~\du   &\du   &\du~\du~\du~\du~\du   &\du   &\du~\du~\du~\du~\du   &\du   &\du~\du~\du~\du~\du   &\du   &\du~\du~\du~\du~\du   \end {array} \right]
\end{align*}

\newpage

\section{A reduced characteristic matrix for the Golay code} \label{S:golay}

The binary Golay code of length $24$ is generated by the following twelve row vectors over $GF(4) = \{ 0,1,a,b \}$ after the concatenation
$0 \mapsto 00, 1 \mapsto 11, a \mapsto 01, b \mapsto 10$. 
\[
\left[ \begin{array}{cccccccccccc}
1 &a &b &1 &b &a &\dt &\dt &\dt &\dt &\dt &\dt \\
a &b &1 &a &1 &b &\dt &\dt &\dt &\dt &\dt &\dt \\ 
\dt &\dt &b &1 &a &1 &b &a &\dt &\dt &\dt &\dt \\
\dt &\dt &a &b &1 &b &a &1 &\dt &\dt &\dt &\dt \\ 
\dt &\dt &\dt &\dt &1 &a &b &1 &b &a &\dt &\dt \\
\dt &\dt &\dt &\dt &a &b &1 &a &1 &b &\dt &\dt \\ 
\dt &\dt &\dt &\dt &\dt &\dt &b &1 &a &1 &b &a \\
\dt &\dt &\dt &\dt &\dt &\dt &a &b &1 &b &a &1 \\ 
b &a &\dt &\dt &\dt &\dt &\dt &\dt &1 &a &b &1 \\
1 &b &\dt &\dt &\dt &\dt &\dt &\dt &a &b &1 &a \\ 
a &1 &b &a &\dt &\dt &\dt &\dt &\dt &\dt &b &1 \\
1 &b &a &1 &\dt &\dt &\dt &\dt &\dt &\dt &a &b
\end{array} \right]
\]
Every even row is a scalar multiple of the preceeding odd row and every odd row is obtained with an obvious symmetry from the preceeding even row. If we apply the concatenation in combination with a transposition of symbols $4k+2$ and $4k+4$, i.e.
\[
\begin{array}{lll}
ab \mapsto 0110 \mapsto 0011,  &b1 \mapsto 1011 \mapsto 1110, &1a \mapsto 1101 \mapsto 1101, \\
ba \mapsto 1001 \mapsto 1100,  &1b \mapsto 1110 \mapsto 1011, &a1 \mapsto 0111 \mapsto 0111,  
\end{array}
\]
then the result is a generator matrix $G$ for the Golay code with optimal minimal spanlength $108 = 12 \cdot 9$. The generator matrix is of the form (1) in \cite{CFV99} and agrees with the generator matrix in (4) of that paper after a permutation $(1\;2)(7\;8)(9\;10)(15\;16)(17\;18)(23\;24).$ 
The reduced characteristic matrix $X$ for $G$ is of the form
\[
X = \left[ \begin{array}{lll}
A &B &C \\ C &A &B \\ B &C &A
\end{array} \right]
\]
with $(A|B|C)$ equal to
{\small
\[
\left[ \begin {array}{cccccccc|cccccccc|cccccccc} 
1&1&0&1&1&1&1&0&1&1&0&0&0&0&0&0&0&0&0&0&0&0&0&0\\
0&1&1&1&1&0&0&0&1&0&0&0&1&0&1&0&1&0&0&0&0&0&0&0\\ 
0&0&1&1&1&1&0&1&1&0&1&1&0&0&0&0&0&0&0&0&0&0&0&0\\ 
0&0&0&1&0&1&1&1&0&1&1&0&0&0&1&0&0&0&1&0&0&0&0&0\\ 
0&0&0&0&1&1&1&0&0&1&1&1&1&1&0&0&0&0&0&0&0&0&0&0\\ 
0&0&0&0&0&1&0&1&0&1&0&0&1&1&0&0&1&0&1&0&1&0&0&0\\ 
0&0&0&0&0&0&1&1&1&0&1&1&0&1&1&1&0&0&0&0&0&0&0&0\\ 
0&0&0&0&0&0&0&1&0&0&1&0&1&0&1&1&0&0&1&0&1&0&1&0
\end{array} \right]
\]
}
The odd rows have spanlength $9$ and the even rows have spanlength $15$.

\section{Nonattacking rooks for function fields} \label{S:rooks}

Minimal spans for characteristic generators have an analogue in discrepancies for function fields.

\bigskip

Given a field $K$ of algebraic functions and two rational points (places), a function in $K$ has span $(i,j)$ if it has a pole of order $j$ at the first point, no other poles, and a zero of order $i$ at the second point. For each $j \in \mathbb{Z}$, let $\sigma(j) \in \mathbb{Z}$ be maximal such that there exists a function with span $(i,j)$. This defines a pair $(\sigma(j),j)$. The pairs fill the plane with an infinite set of nonattacking rooks. let $n > 0$ be minimal such that there exists a function with $(i,j)=(n,n)$. Then the pattern is periodic with period $n$. 

\bigskip

For the function field $F(x,y)$, defined with equation $y^8+y = x^{10}+x^3$ over the field of eight elements, $n=13$. There exist functions with minimal spans 
\begin{multline*}
(i,j) \in \{ (0,0), (1,8), (2,16), (3,10), (4,18), (5,12), (6,20), \\ (7,28), (8,22), (9,30),(10,24),(11,32),(12,40) \}.
\end{multline*}
Figure~6 pictures the minimal spans as nonattacking rooks on a chess board of size $13$
with labeling $\{ 0, 1, \ldots, 11, 12 \} \times \{ 1, 2, \ldots, 12, 0 \}$ modulo $13$. 

\bigskip


\begin{figure}[h]
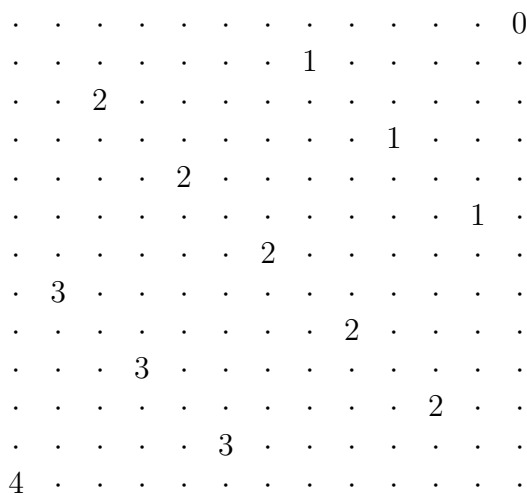
 \label{F:Sz}
\begin{center}
$\begin{array}{ccccccccccccc}
\cdot &\cdot &\cdot &\cdot &\cdot &\cdot &\cdot &\cdot &\cdot &\cdot &\cdot &\cdot &0\\
\cdot &\cdot &\cdot &\cdot &\cdot &\cdot  &\cdot &1 &\cdot &\cdot &\cdot  &\cdot &\cdot \\
\cdot &\cdot &2 &\cdot &\cdot &\cdot &\cdot &\cdot &\cdot &\cdot &\cdot  &\cdot &\cdot \\
\cdot &\cdot &\cdot &\cdot &\cdot &\cdot &\cdot &\cdot &\cdot &1 &\cdot  &\cdot &\cdot \\
\cdot &\cdot &\cdot &\cdot &2 &\cdot &\cdot &\cdot &\cdot &\cdot &\cdot  &\cdot &\cdot  \\
\cdot &\cdot &\cdot &\cdot &\cdot &\cdot &\cdot &\cdot &\cdot &\cdot  &\cdot &1 &\cdot \\
\cdot &\cdot &\cdot &\cdot &\cdot &\cdot &2 &\cdot &\cdot &\cdot &\cdot  &\cdot &\cdot \\
\cdot &3 &\cdot &\cdot &\cdot &\cdot &\cdot &\cdot &\cdot &\cdot &\cdot  &\cdot &\cdot \\
\cdot &\cdot &\cdot &\cdot &\cdot &\cdot &\cdot &\cdot &2 &\cdot &\cdot  &\cdot &\cdot \\
\cdot &\cdot &\cdot &3 &\cdot &\cdot &\cdot &\cdot &\cdot &\cdot &\cdot  &\cdot &\cdot \\
\cdot &\cdot &\cdot &\cdot &\cdot &\cdot &\cdot &\cdot &\cdot &\cdot &2  &\cdot &\cdot \\
\cdot &\cdot &\cdot &\cdot &\cdot &3 &\cdot &\cdot &\cdot &\cdot &\cdot  &\cdot &\cdot \\
4 &\cdot &\cdot &\cdot &\cdot &\cdot &\cdot &\cdot &\cdot &\cdot &\cdot &\cdot &\cdot
\end{array}$
\end{center}
\caption{Nonattacking rooks for the Suzuki curve over the field of eight elements.}
\end{figure}

\end{document}